%% file: TITPaper.tex
\documentclass[journal]{IEEEtran}

\ifCLASSINFOpdf
\else
\fi
	\usepackage{amsmath}
	\usepackage{listings}
	\usepackage{indentfirst}
	\usepackage{graphicx}
	\usepackage{float}
	\usepackage{extarrows}
	\usepackage{amssymb}
	\usepackage{amsthm}
	\usepackage{epstopdf}
	\usepackage{caption}
	\usepackage{subcaption}
	\usepackage{cite,color}
	\usepackage{todonotes}
	\usepackage{tikz,pgfplots}
	\usepackage{xcolor}
	\usetikzlibrary{automata,positioning}
	\usetikzlibrary{arrows,shapes,chains}
	\usepgflibrary{patterns}
	\usepackage{algorithm}
	\usepackage{algorithmic}
	\usepackage{tikz,pgfplots}
	\usetikzlibrary{arrows, intersections}
	\pgfplotsset{plot coordinates/math parser=false}

	\newtheorem{lemma}{Lemma}

	\newtheorem{theorem}{Theorem}
	\newtheorem{definition}{Definition}
	
	\newtheorem{proposition}{Proposition}
	\newtheorem{remark}{Remark}
	
	\usepackage{enumitem}
	\setlist[itemize]{leftmargin=*}
	
\begin{document}
		%
		\title{Age of Information in Random Access Channels}
		%
		%
		%
		
		\author{Xingran~Chen,~\IEEEmembership{Student Member,~IEEE,}
			Konstantinos~Gatsis,~\IEEEmembership{Member,~IEEE,}
			Hamed~Hassani,~\IEEEmembership{Member,~IEEE,}
			and~Shirin~Saeedi~Bidokhti,~\IEEEmembership{Member,~IEEE}
			\thanks{X. Chen, H. Hassani, and S. Saeedi~Bidokhti are with the Electrical and System Engineering Department, University of Pennsylvania, Philadelphia, PA, 19104.\,\,
				E-mail: \{xingranc, hassani, saeedi\}@seas.upenn.edu.}
			\thanks{Konstantinos Gatsis is with the Department of Engineering Science, University of Oxford, OX1~2JD, United Kingdom.\,\,
				E-mail: Konstantinos.gatsis@eng.ox.ac.uk.}
			\thanks{Parts of this work were presented in 2020 IEEE International Symposium on Information Theory. }
		\thanks{This version has been accepted by IEEE  Transactions on Information Theory.}}
		
		%
		%

	\markboth{IEEE Transactions on Information Theory,~Vol.~, No.~,  ~2022}%
	{Chen \MakeLowercase{\textit{et al.}}: IEEE Transactions on Information Theory}
	%



	\maketitle
	
	\begin{abstract}
		In applications of remote sensing, estimation, and control, timely communication is critical but not always ensured by high-rate communication. This work proposes decentralized age-efficient transmission policies for random access channels with $M$ transmitters. We propose the notion of \emph{age-gain} of a packet to quantify how much the packet will reduce the instantaneous age of information at the receiver side upon successful delivery. We then utilize this notion to propose a  transmission policy in which transmitters act in a decentralized manner based on the age-gain of their available packets. In particular, each transmitter sends its latest packet only if its corresponding age-gain is beyond a certain threshold which could be computed adaptively using the collision feedback or found as a fixed value analytically in advance. Both methods improve age of information significantly compared to the state of the art. In the limit of large $M$, we prove that when the arrival rate is  small (below $\frac{1}{eM}$),  slotted ALOHA-type algorithms are order optimal. As the arrival rate increases beyond $\frac{1}{eM}$, while age  increases under slotted ALOHA, it decreases significantly under the proposed age-based policies. For arrival rates $\theta$, $\theta=\frac{1}{o(M)}$, the  proposed algorithms provide a multiplicative gain of at least two compared to the minimum age under slotted ALOHA (minimum over all arrival rates). We conclude that it is beneficial to increase the sampling rate  (and hence the arrival rate) and transmit packets selectively based on their age-gain.  This is surprising and contrary to common practice where the arrival rate is optimized to attain the minimum AoI. We further extend our results to other random access technologies such as Carrier-sense multiple access (CSMA).
	\end{abstract}
	

	%
	\IEEEpeerreviewmaketitle

	\section{Introduction}
	\label{sec:intro}
	\IEEEPARstart{C}{ommunication} networks have witnessed a rapid growth in the past few decades and they have laid a path to the integration of intelligence into  cyber-physical systems, the Internet of Things, smart cities, as well as healthcare systems. Today,  state-of-the-art network communication strategies are  considered reliable and high speed; nevertheless, they often do not perform satisfactorily for time-sensitive applications. For example, in applications of remote sensing, estimation, and control, high-rate communication does not ensure timely communication of data. As a matter of fact, it is often observed that as the capacity of a system is approached, the delay increases significantly and hence so does the age of information.
	
	Age of information (AoI), introduced in \cite{S.Kaul2011-1, S.Kaul2011-2}, measures the freshness of information at the receiver side. AoI is a function of both how often packets are transmitted and how much delay packets experience in the communication network. 
	When the rate of communication is low,  the receiver's AoI will increase (implying that the receiver's information is stale) because the transmitter is not sending packets frequently enough. But even when the transmitter is sending packets frequently, if the system design imposes a large delay for the packets,  the information at the receiver will still be stale.
	The metric of AoI is of great importance in the Internet of Things applications where  timeliness of information is crucial (e.g. in monitoring the status of a system). Another interesting application domain of AoI is in communication for estimation and control \cite{TSSKY2018, YSun2018} where estimation error increases (exponentially) by time before new packets (samples) are received at the destination. It is believed that minimizing AoI is a good proxy for minimizing estimation error \cite{YSun2018, XZMMVWCUM2021, XCXLSSB2020}.

	Assuming  a first come first serve (FCFS) policy,  the work in \cite{KaulYatesGruteser12,YatesKaul19} show in queue theoretic setups that AoI is minimized at an optimal update rate.  Relaxing the restriction of FCFS policies, \cite{Costa16,YatesKaul19} propose packet management policies that discard old packets and improve AoI in wide regimes of operation. This already points to the fact that, under the metric of AoI, rate and reliability have little relevance in the design of timely communication schemes. This is because AoI implicitly assumes that the information content of the packets form a Markov process and hence fresh packets render older packets obsolete. 
	In the past few years, various extensions and  new dimensions have  also been studied in the paradigm of timely communication: source and channel coding were studied in \cite{YatesNajmSoljaninZhong17,Mayekar18,Devassy18,NajmTelatarNasser19}, multi-hop networks were studied in \cite{Buyukates18,8445981,Bedewy19-2}, and scheduling algorithms were studied in \cite{I.Kadota-2018,I.Kadota-2019-1,A.M.Bedewy-2019,Qinghe-2018,Y.-P.Hsu-2018,8807257,8476220}.

	This paper considers the problem of minimizing  age of information over a random access channel. This setup is particularly relevant in remote estimation and control of processes that are observed from decentralized sensors in wireless networks (see also our follow up work \cite{cxrinfocom}). 
	Prior work such as \cite{IKEM2021, I.Kadota-2018, 8807257, I.Kadota-2019-1} consider scheduling policies in multiple access channels that are controlled in a centralized manner.  However, in decentralized (random access) applications, employing such policies would require a huge amount of communication and coordination rendering  them inapplicable.
	Towards designing \emph{decentralized} algorithms for minimizing age of information, \cite{S.K.Kaul-2017,TalakKaramanModiano18} analyze stationary randomized policies under the assumption that  sources generate packets in every time slot (i.e., all sources are active at all times).  Considering the more realistic scenario where packets are generated at random times, \cite{A.Kosta-2019} analyzes round-robin scheduling techniques with and without packet management and also presents partial results for stationary randomized policies.
	Round-robin policies are proved to be age-optimal  in \cite{Z.Jiang-2018-3}  when the number of transmitters are large and the arrival rate is constant. 
	The followup work \cite{Z.Jiang-18-1} additionally assumes that nodes are provided with carrier sensing capabilities and proposes distributed schemes that have good performance in simulations; Nevertheless, \cite{Z.Jiang-18-1} does not address how the parameters of the proposed algorithms should be designed theoretically. The concurrent work \cite{sleepwakeconference, AMBYSRSNBS2021} (published after our work \cite{ISITage}) investigate variants of decentralized age-based schemes for CSMA under energy constraints. In an unslotted, uncoordinated, unreliable multiple access collision channel, \cite{RDYates2020} provides the exact system age and an accurate individual age approximation for a small number of sources. The work \cite{DCAEUOK2020} which was done independently and concurrently to this work considers a threshold-based lazy version of Slotted ALOHA where each transmitter attempts to access the channel with a certain probability when its corresponding age exceeds a certain threshold. Optimizing the threshold and the transmission probabilities are non-trivial and the authors provide analysis only for $M=2$ transmitters for the special case where the arrival rate is equal to $1$.

	In this work, we design decentralized age-based transmission policies  and provide upper and lower bounds on the achievable AoI in interesting regimes of operation. 
	The major part of this paper deals with  random access technologies such as slotted ALOHA that do not assume carrier sensing capabilities. The underlying reason  is threefold: (i) Status packets are generally very short  (as opposed to traditional settings such as streaming where packets are long) and so CSMA is not efficient; (ii) Transmitters have low power capabilities. As such, it is not very efficient (in terms of energy and cost) to perform carrier sensing when the arrival rate is large and CSMA is not particularly useful when the arrival rate is small. More importantly, since transmission power is low, the hidden node problem will be a major issue under CSMA-type protocols; (iii) Our analytical results  are clearer without the additional complexity of CSMA. In Section \ref{sec:extension}, we describe how our findings generalize and apply to CSMA.
	
	The contributions of this paper are as follows. In presenting our results below, we assume large symmetric networks in which we have $M$ transmitters and each transmitter has  arrival rate $\theta$.  The key ideas are summarized in Table~\ref{tabla: ideas}.
	
	\begin{itemize}
		\item We first derive two general lower bounds on AoI for any transmission policy by considering two ideal cases: (i) there is always a fresh packet to be transmitted and hence delivered packets are assumed to experience minimum delay; (ii) all packets are delivered instantaneously upon their arrivals with minimum delay, but without experiencing collisions. The former lower bound turns out to be active as the arrival rate ($\theta$) approaches $1$, and the latter lower bound becomes active when $\theta$ is small, i.e., when the inter-arrival time is the dominant term of the inter-delivery time.
		\item We analyze the well-known slotted ALOHA algorithm. It is known that slotted ALOHA is stable  when the sum arrival rate is below the infamous critical point $\frac{1}{e}$. But it becomes unstable when the sum arrival rate is larger than $\frac{1}{e}$. We prove that when the sum arrival rate is below $\frac{1}{e}$, the normalized age performance of a (stabilized) slotted ALOHA algorithm, properly defined later,  is approximate $\frac{1}{M\theta}$ in the limit of large $M$ and is optimal. 
		We further show numerically that the normalized age performance is close to that of centralized max-weight policies that schedule based on {\it age-gain} (which is formally defined in Section~\ref{sec: Centralized scheduling}) when the sum arrival rate is less than $\frac{1}{e}$. Simulation results show that as the sum arrival rate increases beyond this critical point, the normalized age of slotted ALOHA explodes.
		\item  We then ask if  we can reduce age as the sum arrival rate increases beyond the critical point $\frac{1}{e}$. This is an important question that sheds light on whether increasing the sampling rate is useful when communication is over a random access channel. We find an affirmative answer. We propose two age-based thinning algorithms, i.e., Algorithm~\ref{alg: Updated Threshold-slotted ALOHA} and Algorithm~\ref{alg: Limit Threshold-slotted ALOHA}. The core idea for both algorithms is that transmitters selectively disregard packets in order to mimic an effective (sum) arrival rate equal to $\frac{1}{e}$. In particular, we develop a threshold policy that can be implemented in a decentralized manner at the transmitters and in which packets that offer large age-gains are transmitted and those that offer small age-gains are disregarded. In Algorithm~\ref{alg: Updated Threshold-slotted ALOHA} we propose an adaptive threshold in which the threshold is updated and improved based on the channel feedback. Algorithm~\ref{alg: Limit Threshold-slotted ALOHA} proposes a stationary threshold, in which the threshold is predetermined and thus saves computation costs.  Using Algorithm~\ref{alg: Limit Threshold-slotted ALOHA}, i.e., the stationary thinning method, we prove asymptotically ($M\to\infty$) that for any $\theta$ that is not too small $\left(\theta=\frac{1}{o(M)}\right)$, the normalized age is approximate $\frac{e}{2}$ and twice better than that the minimum age that  (stabilized) slotted ALOHA can attain. Furthermore, numerical results show  that as $\theta$ approaches  $1$, the normalized age approaches $1$ using Algorithm~\ref{alg: Updated Threshold-slotted ALOHA} (the adaptive thinning method) that adaptively optimizes the threshold  in each time slot. Interestingly, we observe that the adaptive thinning algorithm attains a smaller age while increasing the throughput beyond what slotted ALOHA can achieve.
		\item Finally, we generalize our stationary thinning mechanism (Algorithm~\ref{alg: Limit Threshold-slotted ALOHA}), and demonstrate that the idea behind Algorithm~\ref{alg: Limit Threshold-slotted ALOHA} is useful for other random access technologies (e.g. CSMA), see Algorithm~\ref{alg: General Threshold-slotted ALOHA}. In particular, we prove that given a technology that can achieve the throughput $C$ (without coding),  Algorithm~\ref{alg: General Threshold-slotted ALOHA} can attain the normalized age  of $\frac{1}{2C}$.  Numerical results show that it approaches order-optimality in the limit of large $M$. 
	\end{itemize}
	
	\begin{table*}[!htbp]
		\centering
		\resizebox{\textwidth}{12mm}{
			\begin{tabular}{|c|c|c|}
				\hline  
				Algorithms or Bounds & Key ideas& Normalized age performance ($M\to\infty$)\\
				\hline
				Proposition~\ref{pro: lowerbound1} & There is always a fresh packet to be transmitted& Lower bound $\frac{1}{2C_{RA}}$;  $C_{RA}$ is the capacity of the RA channel\\
				\hline
				Proposition~\ref{pro: lowerbound2} & All packets are delivered upon  arrival & Lower bound $\frac{1}{\theta M}$; tight when $\theta < \frac{1}{eM}$\\
				\hline  
				Slotted ALOHA& See details in \cite[Chapter~4.2.3]{B-D.Bertsekas}& Normalized age $\frac{1}{\theta M}$ when $\theta <\frac{1}{eM}$\\
				\hline
				Algorithm~\ref{alg: Updated Threshold-slotted ALOHA}& Adaptive age-based thinning (ALOHA)&
				Decreases age and increases throughput simultaneously\\
				\hline
				Algorithm~\ref{alg: Limit Threshold-slotted ALOHA}& Stationary age-based thinning (ALOHA) & Normalized age $\frac{e}{2}$ for $\theta=\frac{1}{o(M)}$\\
				\hline
				Algorithm~\ref{alg: General Threshold-slotted ALOHA} & Stationary age-based thinning (RA  with maximum throughput $C$)& Normalized age $\frac{1}{2C}$ for $\theta=\frac{1}{o(M)}$\\
				\hline 
		\end{tabular}}
		\caption{Summary of the proposed algorithms and bounds}\label{tabla: ideas}
	\end{table*}

	The rest of the paper is organized as follows.  Section~\ref{sec: System Model and Notation} introduces the system model and notations.  Section~\ref{sec: Lower Bound}  provides lower bounds on NAAoI and Section~\ref{sec: Centralized scheduling} proposes  centralized Max-Weight  scheduling policies to avoid collisions and ensure small NAAoI. Section \ref{sec: Distributed Age-Based Policies} introduces novel decentralized age-based policies and provides asymptotic analysis of their corresponding NAAoI (as $M\to\infty$). In Section~\ref{sec: Numerical results and discussions}, we numerically compare the achievable age of the proposed distributed transmission policies with centralized policies as well as the derived lower bounds and demonstrate that our asymptotic results hold approximately for moderate values of $M$ as well. We finally conclude in Section~\ref{sec: future serach} and discuss future research directions.

	\section{System Model and Notation}
	\label{sec: System Model and Notation}
	We consider a wireless architecture where a controller monitors the status of $M$ identical source nodes over a shared wireless medium. To provide analytical frameworks and closed form solutions, we focus on the symmetric systems (instead of asymmetric ones), and use the profile of all sources as an estimate on an individual source and look at the limit behaviour. Let time be slotted. At the beginning of every slot $k$, $k=1,2,\ldots$, the source node $i$, $i=1,\ldots,M$, generates a new  packet encoding information about its current status with probability $\theta$ and this packet becomes available at the transmitter immediately. We denote this generation/arrival process at the transmitter by $A_i(k)$, where $A_i(k)=1$ indicates that  a new packet is generated at time slot $k$ and $A_i(k)=0$ corresponds to the event where there is no new update. New packets are assumed to replace undelivered older packets at the source (i.e., older packets are discarded), relying on the fact that the underlying processes that are monitored in physical systems are oftentimes Markovian\footnote{We show in Appendix~\ref{App: Proof of first lemma} that this assumption can be made without loss of generality when the performance measure is Age of Information.}.

	The communication media is modeled by a collision channel: If two or more source nodes transmit  at the beginning of the same slot, then the packets interfere with each other (collide) and do not get delivered at the receiver. We use the binary variable $d_i(k)$  to indicate whether a packet is transmitted from source $i$ and received at the destination in time slot $k$. Specifically, $d_i(k)=0$ if source $i$ does not transmit at the beginning of time slot $k$ or if collision occurs; $d_i(k)=1$ otherwise. 
	
	We assume a delay of one time unit in the delivery of packets, meaning that packets are transmitted at the beginning of time slots and, if there is no collision, they are delivered at the end of the same time slot. We  assume that all transmitters are provided with channel collision feedback at the end of each time slot. Specifically, at the end of time slot $k$, $c(k)=1$ if collision happened and $c(k)=0$ otherwise.  In the event that collision occurs,  the involved transmitters can keep the undelivered packets and retransmit them according to their transmission policy (until the packets are successfully delivered or replaced by new packets).

	Our objective is to design \emph{decentralized} transmission mechanisms  to minimize time-average age of information per source node. A decentralized transmission policy is one in which the decision of  transmitter $i$ at time $k$ is dependent only on its own history of actions, the packets arrived so far, $\{A_i(j)\}_{j=1}^k$, as well as the collision feedback received so far, $\{c(j)\}_{j=1}^{k-1}$.
	
	The measure of performance in this work is Age of Information (AoI). Originally defined in  \cite{S.Kaul2011-1, S.Kaul2011-2}, AoI captures the timeliness of information at the receiver side.  We extend the definition a bit further, formally defined below, to also account for the age of information at the source side. Aging at the source/transmitter is caused by the \emph{stochastic nature of arrivals}.
	\begin{definition}\label{def: AoI}
		Consider a source-destination pair. Let $\{k_\ell\}_{\ell\geq 1}$ be the sequence of generation times of packets  and $\{k_\ell'\}_{\ell\geq 1}$ be the sequence of times at which those packets are received at the destination. At any time $\tau$, denote  the index of the last generated packet by $n_s(\tau)=\max\{\ell|k_\ell\leq\tau\}$ and the index of the last received packet by $n_d(\tau)=\max\{\ell|k_\ell'\leq\tau\}$.  The \emph{source's age of information} is defined by $w(k)=k-k_{n_s(k)}$ and the \emph{destination's age of information} is defined by by $h(k)=k-k_{n_d(k)}$.
	\end{definition}
	
	It is clear from the above definition that once there is a new packet available at the transmitter, the older packet(s) cannot contribute to reducing the age of the system. We hence assume without loss of generality that buffers at transmitters are of size 1 and new packets replace old packets upon arrival. 
	We formalize and prove this claim in Appendix~\ref{App: Proof of first lemma}.

	Following Definition~\ref{def: AoI}, let $h_i(k)$ denote the destination's AoI at time slot~$k$ with respect to source~$i$. The age $h_i(k)$ increases linearly as a function of $k$ when there is no packet delivery from source~$i$ and it drops with every delivery to a value that represents how old the received packet is; within our framework, this would be the corresponding source's AoI (in previous time slot) plus 1.
	Without loss of generality, we assume $w_i(1)=0$ and $h_i(1)\geq 0$, and write the recursion of AoI as follows:
	\begin{align}\label{eq: h_i}
		h_i(k)=\left\{\begin{aligned}
			&w_i(k-1)+1&\quad& d_i(k-1)=1\\
			&h_i(k-1)+1&\quad& d_i(k-1)=0
		\end{aligned}\right.
	\end{align}
	and
	\begin{align}\label{eq: w_i}
		w_i(k)=\left\{\begin{aligned}
			&0&\quad& A_i(k)=1\\
			&w_i(k-1)+1&\quad& A_i(k)=0.
		\end{aligned}\right.
	\end{align}
	Note that at the beginning  of each time slot $k$, given the collision feedback $\{c(j)\}_{j\leq k-1}$ and local information about $\{A_i(j)\}_{j\leq k}$,  transmitter $i$ can compute its corresponding source's AoI $\{w_i(j)\}_{j\leq k}$ and destination's AoI $\{h_i(j)\}_{j\leq k}$.

	We define the Normalized Average AoI (NAAoI) as our performance metric of choice\footnote{For any distributed transmission scheme, it is clear that the average AoI increases with the number of source node $M$ for any fixed arrival rate $\theta$. Note that our problem setup allows $M$ to become very large, so to offset the effect introduced by the number of source nodes, we consider the proposed NAAoI.}:
	\begin{equation}\label{eq: EWSAoI-1}
		\small
		\begin{aligned}
			J^{\pi}(M)=\lim_{K\to\infty}\mathbb{E}[J_{K}^{\pi}],\,\, J_{K}^{\pi}=\frac{1}{MK}\sum_{i=1}^{M}\sum_{k=1}^{K}\frac{h_{i}^{\pi}(k)}{M}
		\end{aligned}
	\end{equation}
	where  $\pi$ refers to the underlying transmission policy.

	We consider \emph{centralized policies} and \emph{decentralized age-based policies} in this work. Centralized policies serve as benchmarks. They need a central scheduler who receives information about all arrival processes and previous transmission actions, and  coordinate all the transmitters. When the number of transmitters $M$ gets large, facilitating such scales of coordination is not feasible and we are hence interested in decentralized mechanisms. Randomized policies are easy to implement in a decentralized manner. Previous works \cite{S.K.Kaul-2017,TalakKaramanModiano18} fall into this class but they have the weakness of not utilizing local collision feedback at the transmitters. Utilizing the collision feedback, we aim to make age-based decisions at the transmitters in a decentralized manner.

	\subsection{Notation}
	We use the notations $\mathbb{E}[\cdot]$ and $\Pr(\cdot)$ for expectation and probability, respectively. We denote scalars with lower case letters, e.g. $s$; vectors with underlined lowercase letters, e.g. $\underline{s}$, and matrices with boldface capital letters, e.g. ${\bf S}$. Notation $[\underline{s}]_i$ represents the $i^{th}$ element of $\underline{s}$ and $[{\bf S}]_{ij}$ denotes the element in the $i^{th}$ row and $j^{th}$ column. Random variables are denoted by capital letters, e.g. $S$. We use $M$ to denote the number of transmitters, $K$ to denote the time horizon, and $C$ to denote the capacity of a channel (under a given technology). The operator $(s)^+$ returns $0$ if $s< 0$ and it returns $s$ if $s\geq 0$. $\lfloor s\rfloor$ represents the largest integer $j$ such that $j\leq s$. $O(\cdot)$ and $o(\cdot)$ represent the  Big O and  little o notations according to Bachmann-Landau notation, respectively.   We summarize  the notations in Table~\ref{tabla: notations}.
	\begin{table}[!htbp]
		\centering
		\begin{tabular}{|c|c|}
			\hline  
			$M$& The number of sources\\
			\hline
			$K$ & The time horizon\\
			\hline  
			$\theta$& The generation/arrival rate of new packets\\
			\hline
			$A_i(k)$ & The indicator of the generation/arrival process\\
			\hline
			$d_i(k)$ & The indicator of delivery at source $i$\\
			\hline
			$\lambda_i(k)$ & The indicator of transmission at source $i$\\
			\hline
			$c(k)$ & The indicator of collision in the channel\\
			\hline
			$h_i(k)$ & The destination's AoI at time $k$ w.r.t source $i$ \\
			\hline
			$w_i(k)$ & The source's AoI at time $k$ w.r.t source $i$\\
			\hline
			$\pi$ & A specific transmission and sampling policy\\
			\hline
			$J^\pi(M)$ & Normalized Average AoI with $M$ sources\\
			\hline
			$C_{RA}$ & The sum-capacity of the random access channel\\
			\hline
			$\delta_i(k)$ & The age-gain in time slot $k$ at source $i$\\
			\hline
			$\{\ell_m(k)\}_m$ & The distribution of age-gain in time slot $k$\\
			\hline
			${\tt T}(k)$ & The threshold under the AAT policy at time $k$\\
			\hline
			${\tt T}^*$ & The threshold under the SAT policy\\
			\hline 
		\end{tabular}
		\caption{Useful Notations}\label{tabla: notations}
	\end{table}

	\section{Lower Bound}
	\label{sec: Lower Bound}
	We start by deriving two lower bounds on the achievable age performance. The first lower bound is derived by assuming that there is always a fresh packet to be transmitted (and hence delivered packets are assumed to experience unit-time delays). The second lower bound is derived by assuming that all packets are delivered instantaneously  upon their arrivals (with unit-time delays, but without experiencing collisions). The former is active as $\theta$ approaches $1$ and the latter is active when $\theta$ is small (when the inter-arrival time is the dominant term of the inter-delivery time).

	\begin{figure}[t!]
		\centering
		\includegraphics[width=3in]{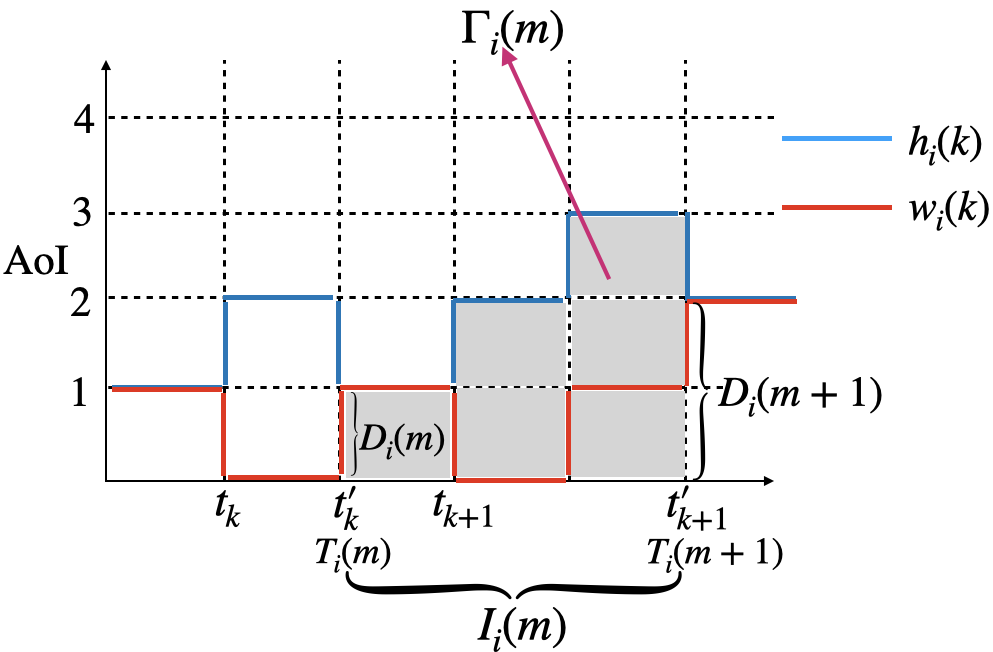}
		\caption{an example of $D_i(m)$, $I_i(m)$, and $\Gamma_i(m)$}
		\label{fig:DI}
	\end{figure}

	Fix a large time horizon $K$ and look at the packets of source $i$. Let $N_i(K)$ denote the number of delivered packets (from source $i$) up to and including time slot $K$. Now consider the $m^{th}$ and $(m+1)^{th}$ \emph{deliveries}  at the receiver and denote the delivery time of them at the receiver by $T_i(m)$ and $T_i(m+1)$, respectively.
	The inter-delivery time $$I_i(m)=T_i({m+1})-T_i(m)$$ is the time between these two consecutive deliveries. Upon arrival of the $m^{th}$ delivered packet at the receiver, the age of information at the receiver drops to the value $D_i(m)$ which represents how much delay the packet has experienced in the system.  Fig.~\ref{fig:DI} illustrates the introduced notation. 
	Now define $\Gamma_i(m)$ as the sum of age functions $h_i(k)$, where $k$ is in the interval $[T_i(m),T_i(m+1))$: 
	\begin{equation}\label{eq:gamma}
		\begin{aligned}
			\Gamma_i(m)=&\sum_{k=T_i(m)}^{T_i(m)+I_i(m)-1}h_i(k)\\
			=&\frac{1}{2}\big(D_i(m)+I_i(m)+D_i(m)\big)\cdot I_i(m)\\
			&-\frac{I_i(m)}{2}\\
			=&\frac{1}{2}I^2_i(m)-\frac{1}{2}I_i(m)+D_i(m)I_i(m).
		\end{aligned}
	\end{equation}
	
	It follows that in the limit of large $K$, we have
	\begin{equation*}
		\small
		\begin{aligned}
			J^{\pi}(M)=\lim_{K\to\infty}\mathbb{E}[J^{\pi}_K]=&\lim_{K\to\infty}\mathbb{E}\left[\frac{1}{M^2}\sum_{i=1}^{M}\frac{1}{K}\sum_{m=1}^{N_i(K)}\Gamma_i(m)\right].
		\end{aligned}
	\end{equation*}
	Using this formulation, we next  lower bound NAAoI. Let $C_{RA}$ denote the sum-capacity of the underlying random access channel. Note that in the limit of large $K$, $\frac{N_i(K)}{K}$ is the throughput of transmitter $i$ and
	\begin{align}\label{eq: sum-capacity}
		\lim_{K\to\infty}\sum_{i=1}^{M}\frac{N_i(K)}{K}\leq C_{RA}.
	\end{align}
	Then, we have the following propositions.
	\begin{proposition}\label{pro: lowerbound1}
		For any transmission policy $\pi$,
		\begin{align*}
			J^{\pi}(M)\geq \frac{1}{2C_{RA}}+\frac{1}{2M}.
		\end{align*}
	\end{proposition}
	\begin{proof}
		The proof is given in Appendix~\ref{App: lowerbound1}.
	\end{proof}
	\begin{proposition}\label{pro: lowerbound2}
		For any transmission policy $\pi$, 
		\begin{align}\label{eq:lbMtheta}
			J^{\pi}(M)\geq\frac{1}{M\theta}.
		\end{align}
	\end{proposition}
	\begin{proof}
		The proof is given in Appendix~\ref{App: lowerbound2}.
	\end{proof}
	Let us give an example of how Proposition \ref{pro: lowerbound1} can be utilized. Note that $C_{RA}$   is not known in general. Nevertheless, any upper bound on $C_{RA}$ gives a  lower bound on the normalized age. Based on \cite{Tsybakov}, the capacity of the random access channel with collision feedback, in the limit of large $M$, is upper bounded by
	$\lim_{M\to\infty}{C}_{RA}\leq 0.568$
	and hence 
	\begin{align}
		\label{eq:lowerboundaloha}
		\lim_{M\rightarrow\infty}J^{\pi}(M)\geq .88.
	\end{align}
	\begin{remark}
		The lower bound in \eqref{eq:lowerboundaloha}  does not assume CSMA capabilities. For CSMA, we have $C_{CSMA}\leq 1$ and hence
		\begin{align}
			\label{eq:lowerboundcsma}
			J^{\pi}(M)\geq \frac{1}{2}+\frac{1}{2M}.
		\end{align}
		We show the asymptotic optimality of this bound  in Section~\ref{sec:extension} as $M\to\infty$.
	\end{remark}

	\section{Centralized Scheduling}\label{sec: Centralized scheduling}
	The first class of schemes that we consider are centralized  schemes that avoid collision by scheduling transmitters one by one. In particular, Max-Weight policies are shown to perform close to optimal in various works such as \cite{I.Kadota-2018,I.Kadota-2019-1,IKEM2021}. Although such schemes are not practical (due to the scale of required coordination), it turns out that they provide useful intuitions and they also serve as a benchmark for comparison in Section \ref{sec: Distributed Age-Based Policies}.
	We assume a central scheduler that can observe all arrival processes and coordinate/control all senders' actions in order to avoid collision.
	
	Denote by $\lambda_i(k)=1$ the event that transmitter $i$ sends a packet and recall that $d_i(k)$ indicates delivery of packets. Note that if $\lambda_j(k)=1$ for another source $j\neq i$, then the packets collide and no packets will be delivered. One can thus write
	\begin{align}\label{eq: d_i}
		d_i(k)=\lambda_i(k)\prod_{j\neq i}\big(1-\lambda_{j}(k)\big).
	\end{align}
	The goal of a central scheduler is to select one source for transmission at each time. Denote $\underline{h}(k)=(h_1(k),h_2(k),\cdots,h_M(k))$. Following the works in \cite{I.Kadota-2019-1, I.Kadota-2018, IKEM2021}, an age-based max-weight  policy can be designed by considering the following  Lyapunov function:
	\begin{align}
		\mathcal{L}(\underline{h}(k))=\sum_{i=1}^{M} h_i(k)
	\end{align}
	and minimizing its corresponding one-step Lyapunov Drift:
	\begin{equation}\label{eq:drift}
		\begin{aligned}
			\Delta(\underline{h}(k))=&\mathcal{L}(\underline{h}(k+1))-\mathcal{L}(\underline{h}(k)).
		\end{aligned}
	\end{equation}
	
	It turns out that the max-weight policy selects, in each time slot $k$, the transmitter that offers the highest \emph{age-gain} $\delta_i(k)$, defined below:
	\begin{align}\label{eq: delta}
		\delta_i(k):=h_i(k)-w_i(k).
	\end{align} 
	$\delta_i(k)$ quantifies how much the instantaneous receiver's age of information reduces  upon successful delivery from transmitter~$i$.
	Proposition \ref{thm: MW} states the above max-weight policy more formally (see Appendix~\ref{APP: Proof of MW} for the proof).
	\begin{proposition}\label{thm: MW}
		For every time slot $k$, define 
		\begin{align}\label{def: MW}
			\ell(k)=\arg\max_i \delta_i(k).
		\end{align} 
		An optimal policy to minimize the one-step drift in \eqref{eq:drift} is to choose $\lambda_{\ell(k)}(k)=1$ and $\lambda_{j}(k)=0$ for all $j\neq \ell(k)$.
	\end{proposition}
	\begin{remark}
		We will show in Section \ref{sec: Distributed Age-Based Policies} how the notion of age-gain plays a central role also in the design of  distributed age-based policies.
	\end{remark}

	\section{Decentralized Age-Based Policies}\label{sec: Distributed Age-Based Policies}

	In this section, we propose a new class of decentralized policies designed to \emph{prioritize} transmissions for the purpose of minimizing age of information. In each time slot $k$, transmitter~$i$ decides whether or not to send its packet depending on its local AoI, and in particular, based on $\delta_i(k)$ \big(defined in \eqref{eq: delta}\big).

	To develop a deeper understanding of our proposed algorithm, let us focus on two regimes of operation assuming large~$M$:
	\begin{itemize}
		\item The regime of infrequent arrivals, where $\theta\leq\frac{1}{eM}$,
		\item The regime of frequent arrivals, where $\theta>\frac{1}{eM}$.
	\end{itemize}
	The choice of these two regimes is made based on the well-established performance of slotted ALOHA with respect to rate (throughput) \cite[Chapter~4]{B-D.Bertsekas}. As explained earlier in Section \ref{sec:intro}, we will first develop our framework for the slotted-ALOHA random access technology and then generalize to other random access technologies in Section \ref{sec:extension}.
	
	The basic idea of slotted ALOHA is  as follows: At every time slot $k$, transmitters send their packets immediately upon arrival unless they are ``backlogged" after a collision in which case they transmit with a backoff probability. In this section, we focus on Rivest's  stabilized slotted ALOHA \cite[Chapter~4]{B-D.Bertsekas}. In this algorithm, all  arrivals are regarded as backlogged nodes that transmit with the backoff probability $p_b(k)$. Let $c(k)=1$ denote the event that collision occurred at time $k$ and $c(k)=0$ denote the complementary event. The backoff probability is calculated through a  pseudo-Bayesian algorithm  based on an {\it estimate} of the number of backlogged nodes $n(k)$ (see \cite[Chapter~4.2.3]{B-D.Bertsekas}):
	\begin{equation}\label{eq: aloha}
		\begin{aligned}
			p_b(k)=&\min\big(1,\frac{1}{n(k)}\big)\\
			n(k)=&\left\{\begin{aligned}
				&n(k-1)+M\theta+(e-2)^{-1} &  & \text{if }c(k)=1\\
				&\max\big(M\theta, n(k-1)+M\theta-1\big) &  & \text{if } c(k)=0.
			\end{aligned}\right.
		\end{aligned}
	\end{equation}
	
	It is well known that this algorithm attains stability of queues for $\theta<\frac{1}{eM}$. In other words,  transmitters can reliably send packets with a sum-rate up to $\frac{1}{e}$ in a decentralized manner \cite[Chapter~4.2.3]{B-D.Bertsekas}.  Asymptotically, when $M\to \infty$,   the probability of delivering a packet in each time slot is $1/e$, the probability of collisions is  $1-2/e$, and the probability of having an idle channel is  $1/e$ (see Appendix~\ref{App: {lem: slotted aloha probabilities}}). Note that when $M\theta<\frac{1}{e}$, the expected total number of delivered packets in every time slot is $M\theta$.

	We  find the asymptotic NAAoI (in the limit of large $M$) in Theorem~\ref{thm: slotted aloha age e} below.  
	
	\begin{theorem}\label{thm: slotted aloha age e}
		Suppose $\theta<\frac{1}{eM}$ and define $$\eta=\lim_{M\rightarrow\infty}M\theta.$$
		Any stabilized slotted ALOHA scheme achieves
		$$\lim_{M\rightarrow\infty}J^{SA}(M) = \frac{1}{\eta}.$$ 	
		Moreover, (stabilized) slotted ALOHA is order optimal in terms of NAAoI. 
	\end{theorem}
	\begin{proof}
		The proof is presented in Appendix~\ref{App: {thm: slotted aloha age e}}. The idea is to divide the sources  into two groups in every time slot $k$:  sources with $\delta_i(k)=0$ and sources with $\delta_i(k)>0$. We show that (i) the contribution of the first group of sources  to  NAAoI is equal to $\frac{1}{M\theta}$, and (ii) the second group constitutes only a vanishing fraction of the nodes and therefore, even though the sources in this group have larger $\delta_i(k)$'s, their total contribution vanishes as $M\to\infty$.
	\end{proof}

	\subsection{Age-Based Thinning}\label{sec: Age-Based Thinning}
	When the arrival rate $\theta$ of each transmitter  approaches $\frac{1}{eM}$, the NAAoI of slotted ALOHA approaches $e$ (see Theorem~\ref{thm: slotted aloha age e}). As $\theta$ increases beyond $\frac{1}{eM}$, the  arrival rate gets larger than the maximum channel throughput ($=e^{-1}$), $n(k)$ overestimates the number of active transmitters and $p_b(k)$ underestimates the optimal probability of transmission, causing the throughput to decrease and the  NAAoI to sharply increase.
	
	Noting that the maximum  channel rate/throughput is $\frac{1}{e}$ when  (stabilized) slotted ALOHA algorithms are applied, a natural question rises: What should the transmitters do in order to ensure a small age of information at the destination when $\theta\geq\frac{1}{eM}$? 
	A naive solution to the above question would be to have each transmitter randomly drop packets and perform at the effective rate $\frac{1}{eM}$. But Theorem~\ref{thm: slotted aloha age e} shows that this only leads to NAAoI $\approx e$ which implies that we will not be able to benefit from the frequency of fresh packets to reduce age.
	
	To benefit from the availability of fresh packets, we devise a decentralized age-based transmission policy in which transmitters prioritize packets that have larger age-gains. In particular, in  each time slot $k$,  transmitters find a common threshold ${\tt T}(k)$ in order to distinguish and keep packets that offer high age-gains. The core idea is to still fully use the channel (depending on the available technology) but to carefully select, in a decentralized manner, what packets to send to minimize age. Recall that $\delta_i(k)$  denotes the age-gain of scheduling  transmitter~$i$. In our proposed algorithms, transmitters that have large age-gains become active and those with small age-gains stay inactive.  More formally, Transmitter $i$ is  called {\it active}  in time slot $k$ if $\delta_i(k)\geq {\tt T}(k)$. Only active transmitters participate in the transmission policy. Alternatively viewed, at time $k$, we propose to discard a fresh packet at transmitter $i$ if $0\leq\delta_i(k)<{\tt T}(k)$ and to keep it otherwise. We refer to this process as  thinning  and   this is done locally at the transmitters based on the  AoI at the source/destination.
	
	Note that no matter how the  transmission policy is designed, since it is decentralized, it may happen that multiple transmitters try to access the channel at the same time, leading to collisions. For simplicity and clarity of ideas, we will restrict attention to slotted ALOHA techniques to resolve such collisions, and in particular the Rivest's stabilized slotted ALOHA\footnote{In classical slotted ALOHA,   ``backlogged nodes'' represent the nodes who have experienced collision and transmit with the backoff probability $p_b(k)$. In our version of Rivest's algorithm, since we have unit buffer sizes, we don't use the term ``backlogged''. We instead work with active nodes. In each time slot $k$, nodes decide based on their local age-gains whether they should be active. Active nodes transmit with probability $p_b(k)$, see \eqref{eq: aloha} where $n(k-1)$ is the number of active nodes in time $k-1$.} described in \eqref{eq: aloha}.
	
	The main underlying challenge is in the design of ${\tt T}(k)$. We propose two algorithms: an adaptive method of calculating ${\tt T}(k)$ for each time slot based on the local collision feedback and a fixed threshold value ${\tt T}^*$ that is found in advance and remains fixed for all time slots $k$. 
	
	In the remainder of this section, we assume that $M$ is large, and $\theta>\frac{1}{eM}$. 
	The following definition comes in handy in presenting our results.
	\begin{definition}
		Consider transmitter $i$ at time slot $k$. If $\delta_i(k)=m$,  we say that  transmitter $i$ is an {\it $m$-order} node. 
		Now let $\ell_m(k)$ be the expected fraction of $m$-order nodes in time slot $k$, i.e.,
		\begin{align}\label{eq: def ell}
			\ell_m(k)=\mathbb{E}\left[\frac{1}{M}\sum_{i=1}^{M}1_{\{\delta_i(k)=m\}}\right].
		\end{align}
		We define $\{\ell_m(k)\}_{m=0}^{\infty}$ as the average {\it node distribution (of the age-gain)} at time $k$.	
	\end{definition}
	
	\subsection{Adaptive Threshold}\label{sec: Updated Threshold-Slotted-ALOHA Schemes}
	
	Let  ${\tt T}(k)$ denote the threshold for decision making in slot $k$. We propose to choose ${\tt T}(k)$ such that it imposes an effective arrival rate equal to $\frac{1}{eM}$ per transmitter. If the effective arrival rate per transmitter is less than $\frac{1}{eM}$, we are not utilizing the channel efficiently. If it is larger, then we are not prioritizing efficiently. This is because we would get a larger pool of packets than slotted ALOHA can support, leading to a reduced throughput and a larger age.  More specifically, we design ${\tt T}(k)$ in three steps:
	\begin{itemize}
		\item[(i)] Compute an estimate of the node distribution of the age-gain;
		\item[(ii)] Find ${\tt T}(k)$ based on the estimated distribution;
		\item[(iii)] Update the estimate of the node distribution based on the chosen ${\tt T}(k)$ and the collision feedback.
	\end{itemize}
	Note that $\{\ell_m(k)\}_{m=0}^{\infty}$ is unknown in decentralized systems. We hence find an estimate of it $\{\hat{\ell}_m(k)\}_{m=0}^{\infty}$ in every time slot. We summarize the process as follows
	\begin{align}\label{eq: recursion-F}
		\{\hat{\ell}_m(k)\}_{m=0}^{\infty}=F(c(k), \{\hat{\ell}_m(k-1)\}_{m=0}^{\infty})
	\end{align}
	where $F(\cdot)$ is a function which will be determined later.
	
	Suppose the estimated node distribution $\{\hat{\ell}_m(k-1)\}_m$ is known at (the end of) time slot $k-1$. We now describe how threshold ${\tt T}(k)$ is designed and how $\{\hat{\ell}_m(k)\}_m$ is updated. For clarity of ideas, let us  view the  time slot $k$ in three stages: 
	The first stage corresponds to the beginning of the time slot when new packets may arrive  and replace the old packets. We denote the time just before the arrival of new packets by $k^-$ and the time just after the arrival of packets by $k^+$.
	After the arrival of new packets, at time $k^+$,  the source's AoI changes from $w_i(k^-)$ to $w_i(k^+)$ and the destination's AoI $h_i(k^+)$ remains the same as $h_i(k^-)$. So the age-gain values  and their node distributions change. We denote the resulting node distribution in this stage by $\{\hat{\ell}_m(k^+)\}_m$.
	In the second stage, transmitters  determine the threshold ${\tt T}(k)$ based on $\{\hat{\ell}_m(k^+)\}_m$. Transmissions happen according to the designed threshold ${\tt T}(k)$. In the third phase, at the end of time slot $k$ when collision feedback is also available, the node distribution is once again estimated. We slightly abuse notation and denote the final estimate of the node distribution at the end of time slot $k$ with $\{\hat{\ell}_m(k)\}_m$.
	The aforementioned three stages of calculating ${\tt T}(k)$ is described next. 
	
	\noindent{\bf Stage 1:} Suppose the estimated node distribution $\{\hat{\ell}_m(k-1)\}_m$ is known at the beginning of slot $k$ before the arrival of new packets. The expected fraction of $m$-order nodes that receive new packets is $\theta \hat{\ell}_m(k-1)$. The order of these nodes increase and this
	changes the expected node distribution to $\{\hat{\ell}_m(k^{+})\}_m$ as a function of $\{\hat{\ell}_m(k-1)\}_m$. Let $a_m(k)$ denote the expected fraction of nodes that have just become $m$-order nodes at time~$k^+$ for $m\geq1$. 
	\begin{lemma}\label{lem: ell_plus}
		The expected fraction of nodes that have just become $m$-order nodes at time $k^+$  is 
		\begin{align}\label{eq: arrival nodes order m}
			a_m(k)=\theta^2\sum_{j=0}^{m-1}\ell_j(k-1) (1-\theta)^{m-j-1}
		\end{align}
		and the expected node distribution of age-gain at time $k^+$ is 
		\begin{equation}\label{eq: recursion of N'}
			\ell_m(k^+)=\left\{\begin{array}{ll}
				(1-\theta)\ell_m(k-1)& m=0\\
				(1-\theta)\ell_m(k-1)+a_m(k)& m\geq 1.
			\end{array}\right.
		\end{equation}
	\end{lemma}
	\begin{proof} The proof is straightforward and delegated to Appendix~\ref{App: {lem: ell_plus}}.\end{proof}
	We define $\hat{a}_m(k)$ as an estimate of $a_m(k)$, which can be obtained by \eqref{eq: arrival nodes order m} and \eqref{eq: recursion of N'} by replacing $\ell_m(k)$, $\ell_m(k^+)$ with $\hat{\ell}_m(k)$, $\hat{\ell}_m(k^+)$, respectively.
	
	\noindent{\bf Stage 2:} The threshold ${\tt T}(k)$ is determined based on $\{\hat{\ell}_m(k^+)\}_m$.  
	We design ${\tt T}(k)$ such that the \emph{effective arrival rate} of packets that have an age-gain above ${\tt T}(k)$  is close to $\frac{1}{e}$. In other words, we \emph{thin} the arrival process using  local age information. The critical point $\frac{1}{e}$ is the maximum sum arrival rate that ALOHA can support. So if the effective sum arrival rate falls below $\frac{1}{e}$, we do not use the full channel capacity\footnote{Here, capacity refers to the maximum achievable sum rate under ALOHA.} and if we operate above $\frac{1}{e}$, then we incur additional collisions and delay.

	We define the effective arrival rate as the fraction of sources with new arrivals whose age-gain is larger than or equal to ${\tt T}(k)$. 
	Recall that $\hat{a}_m(k)$  is the estimation of the expected fraction of nodes that have just become $m$-order nodes at time $k^+$ (coming from lower order nodes). So the total (estimated) fraction of nodes whose age-gain would, for the first time, pass the threshold ${\tt T}(k)$ is
	\begin{align*}
		\sum_{m\geq {\tt T}(k)}\hat{a}_m(k).
	\end{align*}

	We propose to choose ${\tt T}(k)$ according to the following rule:
	\begin{equation}\label{eq: threshold}
		\footnotesize
		\begin{aligned}
			{\tt T}(k)=\max\left\{t|\sum_{m\geq t}\hat{a}_m(k)\geq\frac{1}{eM}\right\}.
		\end{aligned}
	\end{equation}
	\begin{remark}
		We chose ${\tt T}(k)$ to be the maximum threshold value that does not bring effective sum arrival rate below $\frac{1}{e}$. This is due to the integer nature of age and hence $k$. One can also time share between ${\tt T}(k)-1$ and ${\tt T}(k)$ to operate at an effective sum arrival rate (almost) equal to $\frac{1}{e}$. Thus, to simplify \eqref{eq: aloha}, we can replace $M\theta$ by the effective arrival rate $\frac{1}{e}$ in \eqref{eq: aloha}.
	\end{remark}
	
	\begin{remark}
		The threshold ${\tt T}(k)$ in \eqref{eq: threshold} can also be applied to the regime $0<\theta<\frac{1}{eM}$. In this regime, $\sum_{m=1}^{\infty}a_m(k)<\frac{1}{eM}$. Therefore, assuming that the estimates $\hat{a}_m(k)$ are accurate, the threshold is ${\tt T}(k)\leq 1$, reducing the proposed algorithm to the slotted ALOHA.
	\end{remark}

	\noindent{\bf Stage 3:} Once the threshold ${\tt T}(k)$ is determined, each transmitter verifies locally if its age-gain is above the specified threshold. If so, it transmits its packet with probability $p_b(k)$ defined in \eqref{eq: aloha} mimicking slotted ALOHA.  If collision happens or if all nodes abstain from transmitting, then AoI at the destination increases by $1$ for all sources. If only one node transmits, its packet will be delivered successfully and the corresponding age at the destination drops to the source's AoI.

	\subsection{Estimating the node distribution}\label{sec: Estimating the node distribution}
	It remains to estimate $\hat{\ell}_m(k)$ at the end of time slot $k$, which will serve in computing ${\tt T}(k+1)$ in the next time slot.
	We assume that at the end of time slot $k$, all transmitters are provided with collision feedback from the channel and we hence consider  two cases separately: $c(k)=0$ and $c(k)=1$.
	
	If collision has occurred, i.e., $c(k)=1$, then the order of nodes will not change:
	\begin{equation}\label{eq: updated-recursion N-1}
		\begin{aligned}
			\hat{\ell}_m(k)=\hat{\ell}_m(k^+),\quad m\geq0.
		\end{aligned}
	\end{equation}

	If there was no collision, i.e., $c(k)=0$, then either a packet was delivered or no packet was delivered.  Recall that we design ${\tt T}(k)$ to impose (in the limit of large $M$) an effective sum arrival rate almost equal to $\frac{1}{e}$. Following Lemma~\ref{lem: slotted aloha probabilities} in Appendix \ref{App: {lem: slotted aloha probabilities}},  the two events of idle and successful delivery are almost equiprobable for large $M$:
	\begin{align*}
		&\lim_{k\to\infty}\Pr\big(\sum_{i=1}^{M}d_i(k)=1, c(k)=0\big)\approx\frac{1}{e}\\
		&\lim_{k\to\infty}\Pr\big(\sum_{i=1}^{M}d_i(k)=0, c(k)=0\big)\approx\frac{1}{e}.
	\end{align*}
	Thus, condition on $c(k)=0$, a packet is delivered with probability $1/2$, i.e., the  expected number of delivered packet is $1/2$ and by the inherent symmetry of the system, each active node has the same chance to deliver a new packet. For any $m\geq {\tt T}(k)$, a packet is delivered by $m$-order nodes with probability 
	\begin{align}\label{eq: r_m(k)}
		r_m(k)=\frac{\ell_m(k^+)}{\sum_{t\geq {\tt T}(k)}\ell_t(k^+)}.	
	\end{align}
	The expected number of $m$-order nodes is $M\ell_m(k^+)$ and the expected number of delivered packets by $m$-order nodes (condition on $c(k)=0$) is $\frac{r_m(k)}{2}$. Note that $m$-order nodes can not deliver more than $M\ell_m(k^+)$ packets since the total number of $m$-order nodes is $M\ell_m(k^+)$ and the buffer size is $1$, then
	\begin{align}\label{eq: rm less than ell}
		\frac{r_m(k)}{2}<M\ell_m(k^+).	
	\end{align}
	In order to estimate the expected fraction of $m$-order nodes that have a successful delivery, we simply plug in $\hat{\ell}_m(k^+)$ as an estimate for $\ell_m(k^+)$. Since \eqref{eq: rm less than ell} does not necessarily hold anymore using the estimates, we estimate the expected fraction of $m$-order nodes with a successful delivery as follows: 
	\begin{align*}
		\frac{1}{M}\min\big(\frac{r_m(k)}{2},M\hat{\ell}_m(k^+)\big),
	\end{align*}
	where $r_m(k)$ is computed by \eqref{eq: r_m(k)} and replacing $\ell_m(k)$ with $\hat{\ell}_m(k)$. Consequently, the update rule of the node distribution of age, $\{\hat{\ell}_m(k)\}_m$, is given as follows:
	\begin{equation}\label{eq: update-recursion N-2}
		\begin{aligned}
			\hat{\ell}_0(k)=&\hat{\ell}_0(k^+)+\sum_{m={\tt T}(k)}^{\infty}\min\big(\frac{r_m(k)}{2M},\hat{\ell}_m(k^+)\big)\\
			\hat{\ell}_m(k)=&\hat{\ell}_m(k^+),\ \qquad\qquad\qquad\qquad 1\leq m\leq {\tt T}(k)-1\\
			\hat{\ell}_m(k)=&\left(\hat{\ell}_m(k^+)-\frac{r_m(k)}{2M}\right)^+,\quad\quad m\geq {\tt T}(k).
		\end{aligned}
	\end{equation}
	Collecting Stages~1 - 3, from \eqref{eq: recursion of N'}, \eqref{eq: recursion-F} can be re-written as
	\begin{equation}\label{eq: recursion-F1}
		\begin{aligned}
			\hat{\ell}_0(k)=&(1-\theta)\hat{\ell}_0(k-1)+1_{\{c(k)=0\}}\\
			&\times\sum_{m={\tt T}(k)}^{\infty}\min\big(\frac{r_m(k)}{2M},(1-\theta)\hat{\ell}_m(k-1)+\hat{a}_m(k)\big)\\
			\hat{\ell}_m(k)=&(1-\theta)\hat{\ell}_m(k-1)+\hat{a}_m(k),\ \quad1\leq m\leq {\tt T}(k)-1\\
			\hat{\ell}_m(k)=&\left((1-\theta)\hat{\ell}_m(k-1)+\hat{a}_m(k)-1_{\{c(k)=0\}}\frac{r_m(k)}{2M}\right)^+,\\
			&m\geq {\tt T}(k).	
		\end{aligned}
	\end{equation}
	where $a_m(k)$ and $r_m(k)$ are defined in \eqref{eq: arrival nodes order m} and \eqref{eq: r_m(k)}, respectively. Finally, in this case, the probability of transmission is calculated by \eqref{eq: aloha}, where $M\theta$ is replaced by the effective arrival rate $\frac{1}{e}$.

	Algorithm~\ref{alg: Updated Threshold-slotted ALOHA}
	describes the proposed distributed age-based transmission policy. We numerically evaluate its age performance in Section \ref{sec: Numerical results and discussions} and analyze a stationary version of it when the threshold is fixed in Section \ref{sec: limiting threshold}.   Comparing with the slotted ALOHA in \eqref{eq: aloha}, Algorithm~\ref{alg: Updated Threshold-slotted ALOHA} significantly reduces the NAAoI  when the sum arrival rate is beyond $\frac{1}{e}$ (see Fig.~\ref{fig-age2_sep1} and Fig.~\ref{fig-age6_sep1}). It achieves this by carefully selecting and delivering packets with a large age-gain. The NAAoI under Algorithm~\ref{alg: Updated Threshold-slotted ALOHA} decreases sharply when the arrival rate $\theta$ approaches  $1$ (see Figure~\ref{fig-age1}). In particular, the NAAoI it achieves at $\theta=1$ is almost $1$.  Contrasting that with the lower bound in Proposition \ref{pro: lowerbound1}, one comes to the conclusion that the throughput achieved by Algorithm~\ref{alg: Updated Threshold-slotted ALOHA} is larger than that of a standard slotted ALOHA. This is because of the implicit coordination that is facilitated by estimating and utilizing the age gain distributions for decision making.
	
	\begin{algorithm}
		\caption{Adaptive Age-based Thinning (AAT)}\label{alg: Updated Threshold-slotted ALOHA}
		\begin{algorithmic}
			
			\STATE {Set a large integer $N$ and the time horizon $K$.}
			
			\STATE {Set initial points:  $h_i(0)=1$, $w_i(0)=0$ for $i=1,2,\cdots,M$; $c(0)=0$; ${\tt T}(0)=1$; $p_b(0)=1$; $n(0)=0$; $k=1$. }
			
			\REPEAT 
			
			\STATE {\bf Step 1:} Calculate $\{\hat{\ell}_m(k^{+})\}_{m=1}^{N}$ by \eqref{eq: recursion of N'}.

			\STATE {\bf Step 2:} Calculate ${\tt T}(k)$ by \eqref{eq: threshold}.
			
			\STATE {\bf Step 3:} For transmitter $i$, $i=1,\ldots,M$: compute $\delta_i(k^+)=h_i(k^+)-w_i(k^+)$;  if $\delta_i(k^+)<{\tt T}(k)$, then it does not transmit packets; if $\delta_i(k^+)\geq {\tt T}(k)$, then it transmits a packet with probability $p_b(k)$\footnotemark.
			
			\STATE {\bf Step 4:} If $c(k)=0$,  calculate $\{\hat{\ell}_m(k)\}_{m=1}^{N}$ by \eqref{eq: updated-recursion N-1}, and  if $c(k)=1$, 
			calculate $\{\hat{\ell}_m(k)\}_{m=1}^{N}$ by \eqref{eq: update-recursion N-2}. Calculate $p_b(k+1)$ by \eqref{eq: aloha} in which  $M\theta$ is replaced by $\min(M\theta, e^{-1})$.

			\UNTIL{$k=K$}
			
			\STATE Calculate $J_K^{AAT}$ by \eqref{eq: EWSAoI-1}.
		\end{algorithmic}
	\end{algorithm}

	\begin{remark}
		From \eqref{eq: aloha}, to estimate the number active nodes in each time slot, the number of nodes in the network is needed. We set $M$ to be a pre-determined parameter, which is known to all nodes.  
	\end{remark}
	
	\begin{remark}
		\label{remarkinaccuracies}
		The estimates $\{\hat{\ell}_m(k)\}_m$ and $\{\hat{a}_m(k)\}_m$ in Algorithm~\ref{alg: Updated Threshold-slotted ALOHA} are not exactly accurate and this is due to the integer nature of the threshold. Assume that $\{\hat{\ell}_m(k_0)\}_m$ and $\{\hat{a}_m(k_0)\}_m$ are exactly accurate in time slot $k_0$. We may have $\sum_{m\geq{\tt T}(k_0)+1}\hat{a}_m(k_0)<\frac{1}{eM}$ but $\sum_{m\geq{\tt T}(k_0)}\hat{a}_m(k_0)>\frac{1}{eM}$ in which case the effective arrival rate \big($=\sum_{m\geq{\tt T}(k_0)}\hat{a}_m(k_0)$\big) would be larger than $e^{-1}$. But the steps in \eqref{eq: update-recursion N-2} are derived by assuming an effective arrival rate $\frac{1}{e}$ and this leads to  inaccuracies in our estimates $\{\hat{\ell}_m(k)\}_m$ and $\{\hat{a}_m(k)\}_m$ as computed in Algorithm~\ref{alg: Updated Threshold-slotted ALOHA}. 
	\end{remark}

	\begin{remark}
		We updated $\{\hat{\ell}_m(k)\}_m$ as a function of $\{\hat{\ell}_m(k-1)\}_m$ and the collision feedback $c(k)$, hence the name adaptive. ${\tt T}(k)$ and $\{\hat{\ell}_m(k)\}_m$ are known at all sources and every source finds the same ${\tt T}(k)$. If we update $\{\ell_m(k)\}_m$ (not $\{\hat{\ell}_m(k)\}_m$) by the conditional expectation of $\{\ell_m(k)\}_m$, condition on $\{\ell_m(k^+)\}_m$ but not on $c(k)$, we will find a fixed limiting threshold ${\tt T}^*$ discussed next. 
	\end{remark}
	
	\footnotetext{A packet with a large age-gain must have a packet ready to transmit.}

	\subsection{Fixed Threshold}\label{sec: limiting threshold}
	A simple variant of the age-based thinning method is found when the threshold ${\tt T}(k)={\tt T}^*$ is fixed throughout the transmission phase.
	In particular, we design ${\tt T}^*$ ahead of time based on the node distribution in the stationary regime. By doing so, we cannot benefit from the collision feedback to adaptively choose ${\tt T}(k)$. However, this framework is preferable for deriving analytical results. 
	
	We  use the  framework and derivation we developed for adaptive thinning in order to find a fixed ``optimal" ${\tt T}^*$ that imposes an effective arrival rate approximately\footnote{This approximation is due to the integer nature of the age threshold.} equal to $1/e$. Note that a larger arrival rate implies further random thinning of the packets to meet the fundamental rate $1/e$ (as opposed to the selective nature of thinning by imposing an age thresholding) and a smaller arrival rate corresponds to inefficient utilization of the channel. 
	
	The  major difference between an adaptive threshold and a fixed threshold is  in the update rules \eqref{eq: updated-recursion N-1}-\eqref{eq: update-recursion N-2} because $c(k)$ is not known when ${\tt T}^*$ is designed. In particular, the update rule \eqref{eq: updated-recursion N-1}-\eqref{eq: update-recursion N-2} is replaced by an average rule that weighs $c(k)=1$ with probability $1-\frac{2}{e}$ and $c(k)=0$  with probability $\frac{2}{e}$ (following Lemma~\ref{lem: slotted aloha probabilities}).

	By the stationarity of the scheme,  the limit of $\{\ell_m(k)\}_{m=0}^{\infty}$ and $\{\ell_m(k^+)\}_{m=0}^{\infty}$ exist as $k\rightarrow\infty$. Denote the two limits by  $\{\ell_m^*\}_{m=0}^{\infty}$ and $\{\ell_m^{+*}\}_{m=0}^{\infty}$, respectively. 
	Similar with \eqref{eq: recursion of N'}, the update rule of Stage~1 implies
	\begin{equation}\label{eq: limit-threshold1}
		\begin{aligned}
			&\ell_0^{+*}=(1-\theta)\ell_0^*\\
			&\ell_m^{+*}=(1-\theta)\ell_m^*+a_m^*\quad\qquad m\geq 1
		\end{aligned}
	\end{equation}
	where 
	\begin{align}a_m^*=\theta^2\sum_{j=0}^{m-1}\ell_j^* (1-\theta)^{m-j-1}\quad\qquad m\geq1.\label{def-astar}\end{align}
	Since we let ${\tt T}(k)={\tt T}(k-1)={\tt T}^*$, the threshold proposed in Stage~2 is
	\begin{align}\label{eq: limit-threshold2}
		{\tt T}^*=\max\Big\{t|\sum_{m\geq t}a_m^*\geq\frac{1}{eM}\Big\}.
	\end{align}
	
	Next, consider Stage~3. In contrast to Section \ref{sec: Estimating the node distribution}, we do not utilize collision feedback in finding ${\tt T}(k)$. So estimating the fraction of $m$-order nodes at the end of time slot $k$ will account for $c(k)=1$ with probability $1-\frac{2}{e}$ and $c(k)=0$ with probability $\frac{2}{e}$ (see Lemma~\ref{lem: slotted aloha probabilities}). We hence obtain
	\begin{equation}\label{eq: limit-threshold3}
		\begin{aligned}
			&\ell_0^*=\ell_0^{+*}+ \frac{1}{eM}\\
			&\ell_m^*=\ell_m^{+*},\qquad\qquad\qquad 1\leq m\leq {\tt T}^*-1\\
			&\ell_m^*=\ell_m^{+*}- \frac{r_m^*}{eM},\quad\qquad m\geq {\tt T}^*
		\end{aligned}
	\end{equation}
	where
	\begin{align*}
		r_m^*=\ell_m^{+*}/\sum_{i={\tt T}}^{\infty}\ell_i^{+*}.
	\end{align*}
	Putting together \eqref{eq: limit-threshold1} - \eqref{eq: limit-threshold3},  we obtain
	\begin{equation}\label{eq: limit collection}
		\begin{aligned}
			&\ell_0^*=(1-\theta)\ell_0^*+\frac{1}{eM}\\
			&\ell_m^*=(1-\theta)\ell_m^*+a_m^*,\quad 1\leq m\leq {\tt T}^*-1\\
			&\ell_m^*=(1-\theta)\ell_m^*+a_m^*-\frac{r_m^*}{eM},\quad m\geq {\tt T}^*
		\end{aligned}
	\end{equation}
	and  conclude the following lemma (see  Appendix~\ref{App: {thm: N', E_m}} for the proof).
	\begin{lemma}\label{thm: N', E_m}
		As $k\to\infty$, the stationary distributions $\{\ell_m^*\}_{m}$, $\{\ell_m^{+*}\}_{m}$ and $\{a_m^*\}_{m}$ satisfy the following properties:
		\begin{align}
			&\ell_m^*=\left\{\begin{array}{ll}\frac{1}{eM\theta}&m=0\\
				\frac{1}{eM}& 1\leq m\leq{\tt T}^*-1\end{array}\right.\\
			&\ell_m^{+*}=\frac{1}{eM}\quad\ \  1\leq m\leq{\tt T}^*-1\\
			&a_m^*=\frac{\theta}{eM}\quad\,\  1\leq m\leq {\tt T}^*\label{eq: a_mstar}.
		\end{align}
	\end{lemma}
	The closed form expression of the fixed threshold $\tt T^*$ is given below (see  Appendix~\ref{App: {thm: limit threshold}} for the proof) and Algorithm \ref{alg: Limit Threshold-slotted ALOHA}  describes our stationary ge-based transmission policy.
	\begin{theorem}\label{thm: limit threshold}
		The fixed threshold ${\tt T}^*$ in \eqref{eq: limit-threshold2} has the following closed form expression: $${\tt T}^*=\lfloor eM-\frac{1}{\theta}+1 \rfloor.$$
	\end{theorem}
	\begin{remark}
		The threshold in Theorem~\ref{thm: limit threshold} can be applied to the regime $0<\theta<\frac{1}{eM}$ as well. In particular, in this regime, the threshold is  ${\tt T}^*\leq 0$  and the proposed algorithm  reduces to the slotted ALOHA.
	\end{remark}
	
	\begin{algorithm}
		\caption{Stationary Age-based Thinning (SAT)}\label{alg: Limit Threshold-slotted ALOHA}
		\begin{algorithmic}
			\STATE {Set the time horizon $K$.}
			
			\STATE {Set initial points:  $h_i(0)=1$, $w_i(0)=0$ for $i=1,2,\cdots,M$; $c(0)=0$; $p_b(0)=1$; $n(0)=0$; $k=1$. }

			\STATE {Calculate ${\tt T}^*= \lfloor eM-\frac{1}{\theta}+1 \rfloor$.} 
			
			\REPEAT
			\STATE {\bf Step 1:} For transmitter $i$, compute $\delta_i(k)=h_i(k)-w_i(k)$,  if $\delta_i(k)<{\tt T}^*$, then it does not transmit packets; if $\delta_i(k)\geq {\tt T}^*$, then it transmits a packet with probability $p_b(k)$.
			
			\STATE {\bf Step 2:} Calculate $p_b(k+1)$ by \eqref{eq: aloha} in which $M\theta$ is replaced by $\min(M\theta,e^{-1})$.

			\UNTIL{$k=K$}
			
			\STATE Calculate $J_K^{SAT}$ by \eqref{eq: EWSAoI-1}.
		\end{algorithmic}
	\end{algorithm}

	We finally prove asymptotically (as $M\to\infty$) that the Stationary Age-based Thinning (SAT) policy described in Algorithm~\ref{alg: Limit Threshold-slotted ALOHA} significantly reduces age when $1/\theta=o(M)$. Recall that at $\theta=\frac{1}{eM}$, we have $\lim_{M\to\infty}J^{SA}(M)=e$. For larger arrival rates $\theta$ where $1/\theta=o(M)$, we prove that Algorithm~\ref{alg: Limit Threshold-slotted ALOHA} sharply reduces AoI from $e$ to $\frac{e}{2}$. 
	
	\begin{theorem}\label{thm: theta=1, e/2}
		For any $\theta = \frac{1}{o(M)}$, 
		\begin{align*}
			\lim_{M\rightarrow\infty}J^{SAT}(M)=\frac{e}{2}.
		\end{align*}
	\end{theorem} 
	\begin{remark}
		The minimum NAAoI attained by a stabilized slotted ALOHA is (asymptotically) $e$ and it is achieved at $\theta=\frac{1}{eM}$ (See Theorem \ref{thm: slotted aloha age e} and Fig. \ref{fig-age6_sep1}). Theorem~\ref{thm: theta=1, e/2} shows that our proposed SAT policy attains the NAAoI $\frac{e}{2}$ (asymptotically) for $\theta=\frac{1}{o(M)}$. This provides a multiplicative factor of $2$ compared to the minimum NAAoI under slotted ALOHA. Moreover, simulation results show that the AAT policy outperforms the SAT policy for $\theta=\frac{1}{o(M)}$ (see Fig. \ref{fig-age6_sep1}). \end{remark}
	
	\begin{proof}The proof is given in Appendix~\ref{App: proof of Theorem thm: theta=1, e/2}. Here, we provide the road-map of the proof. In every time slot $k$, the sources can be divided into two groups: 1) sources with $\delta_i(k)<{\tt T^*}$; 2) sources with $\delta_i(k)\geq{\tt T^*}$. The first group of sources have the main contribution to $J^{SAT}(M)$ (which is equal to $\frac{e}{2}$) when $M\rightarrow\infty$. The contribution of the second group of sources to $J^{SAT}(M)$ vanishes when $M\rightarrow\infty$.
	\end{proof}

	\subsection{Extensions to Other Random Access Technologies}\label{sec:extension}
	So far, we restricted attention to slotted ALOHA as the main random access technology. However, in the past decade, novel technologies such as Carrier Sensing Multiple Access (CSMA) technologies have emerged and led to significant improvements in terms of throughput. It is interesting to know how they perform with regard to age, especially since they are known to have large delays \cite{PKHXL2014,LJMLJNRSJW2012,YYTY2003, AMMAAE2020, sleepwakeconference, MWYD2019, AMBYSRSNBS2021, LK1975}.  In this regard, \cite{AMBYSRSNBS2021, sleepwakeconference} have proposed an efficient sleep-wake mechanism for wireless networks that attains the optimal trade-off between minimizing the AoI and energy consumption. In \cite{AMMAAE2020}, a network with $M$ sources (links) under CSMA scheme was considered and the closed form of average age of information was derived as a function of the back-off time and generation rate. In \cite{MWYD2019}, the notion of broadcast age of information  was investigated in wireless networks with CSMA/CA technologies.

	In this section, we outline how the age-based thinning method described in Section~\ref{sec: limiting threshold} (with a fixed threshold) can be applied to other random access technologies. 
	For this purpose, we consider any transmission policy $\pi$ that does not employ coding across packets. All existing collision avoidance and resolution techniques such as  ALOHA and CSMA \cite{Abr70, KLLS1975, Gallager1978, Mosely1979} fall into this class. Now develop a variant of the transmission policy $\pi$ in which only the most recent packets of each transmitter are preserved and all older packets are discarded.  Denote this policy by $\pi^{(1)}$. Define $C^{\pi^{(1)}}(M)$ as the maximum sum throughput when applying the transmission policy $\pi^{(1)}$ in a system with $M$ sources,  and denote the limit, when $M\rightarrow\infty$, by $C^{\pi^{(1)}}$.
	Consider the age-based thinning process  in two steps: (i)  the threshold ${\tt T}^*$ is calculated, (ii) all nodes with  age-gains larger than or equal to ${\tt T}^*$ become active and transmit using the prescribed random access technology\footnote{Here, ``prescribed random access technology'' refers to the specific transmission scheme which is applied to the random access channel.}.  
	
	Consider $M$ to be large, and suppose the expected number of delivered packets per time slot is around $\min(M\theta,C^{\pi^{(1)}})$. Therefore, \eqref{eq: limit-threshold1} remains the same and \eqref{eq: limit-threshold2}  takes the following form:
	\begin{align}\label{eq: general limit-threshold2}
		{\tt T}^*=\max\left\{t|\sum_{m\geq t}a_m^*\geq\frac{C^{\pi^{(1)}}}{M}\right\}.
	\end{align}
	Following a similar argument as in Section \ref{sec: limiting threshold}, the equations in \eqref{eq: limit-threshold3} can be written more generally as follows:
	\begin{equation}\label{eq: general limit-threshold3}
		\begin{aligned}
			&\ell_0^*=\ell_0^{+*}+\min(\theta, \frac{C^{\pi^{(1)}}}{M})\\
			&\ell_m^*=\ell_m^{+*}\qquad\qquad\qquad \qquad\,\,\,\,\,\,\, 1\leq m\leq {\tt T}^*-1\\
			&\ell_m^*=\ell_m^{+*}-r_m^*\min(\theta, \frac{C^{\pi^{(1)}}}{M})\quad m\geq {\tt T}^*
		\end{aligned}
	\end{equation}
	where
	\begin{align*}
		r_m^*=\ell_m^{+*}/\sum_{i={\tt T}^*}^{\infty}\ell_i^{+*}.
	\end{align*}
	Combining \eqref{eq: limit-threshold1}, \eqref{eq: general limit-threshold2}, \eqref{eq: general limit-threshold3}, we thus find
	\begin{align}
		&\ell_0^*=\left\{\begin{array}{ll}\min(1, \frac{C^{\pi^{(1)}}}{M\theta})&m=0\\
			\min(\theta,\frac{C^{\pi^{(1)}}}{M})& 1\leq m\leq{\tt T}^*-1\end{array}\right.\label{eq: l_m+C1}\\
		&\ell_m^{*}=\min(\theta,\frac{C^{\pi^{(1)}}}{M}) \quad\ \ \,1\leq m \leq {\tt T}^*-1\label{eq: l_m+C2}\\
		&a_m^*=\min(\theta^2,\frac{\theta C^{\pi^{(1)}}}{M})\quad\ 1\leq m\leq {\tt T}^*.\label{eq: l_m+C3}
	\end{align}
	Moreover, the  threshold $T^*$ takes a simple closed-form expression as stated below (and proved in Appendix~\ref{App: {thm: general limit threshold}}).
	\begin{theorem}\label{thm: general limit threshold}
		The fixed threshold ${\tt T}^*$ in \eqref{eq: general limit-threshold2} has the following closed form expression: $${\tt T}^*=\left\lfloor \frac{M}{C^{\pi^{(1)}}}-\frac{1}{\theta}+1 \right\rfloor.$$
	\end{theorem}
	Using this result, Algorithm~\ref{alg: General Threshold-slotted ALOHA}  proposes a decentralized age-based thinning method that can be applied to any given stationary random access technology. 
	\begin{algorithm}
		\caption{Generalized Stationary Age-based Thinning (GSAT)}\label{alg: General Threshold-slotted ALOHA}
		\begin{algorithmic}
			\STATE {Set the time horizon $K$.}
			
			\STATE {Set initial points:  $h_i(0)=1$, $w_i(0)=0$ for $i=1,2,\cdots,M$; $c(0)=0$; $p_b(0)=1$; $n(0)=0$; $k=1$. }
			
			\STATE {Calculate the threshold ${\tt T}(C^{\pi^{(1)}})= \lfloor \frac{M}{C^{\pi^{(1)}}}-\frac{1}{\theta}+1 \rfloor$.}
			\REPEAT
			\STATE For the source node $i$, compute $\delta_i(k)=h_i(k)-w_i(k)$. If $\delta_i(k)<{\tt T}(C^{\pi^{(1)}})$ remain silent; If $\delta_i(k)\geq {\tt T}(C^{\pi^{(1)}})$,  transmits according to the random access technology $\pi^{(1)}$.
			\UNTIL{$k=K$}
			
			\STATE Calculate $J_K^{GSAT}$ by \eqref{eq: EWSAoI-1}.
		\end{algorithmic}
	\end{algorithm}
	
	We prove an analogue to Theorem \ref{thm: theta=1, e/2}, showing that the Generalized Stationary Age-based Thinning policy (GSAT) proposed in Algorithm~\ref{alg: General Threshold-slotted ALOHA} reduces age  to $\frac{1}{2C^{\pi^{(1)}}}$  as $\theta$ increases.
	\begin{theorem}\label{thm: theta=1, xi/2}
		For any $\theta=\frac{1}{o(M)}$, $$\lim_{M\rightarrow\infty}J^{GSAT}(M)=\frac{1}{2C^{\pi^{(1)}}}.$$
	\end{theorem}
	\begin{proof}The proof of Theorem~\ref{thm: theta=1, xi/2} is given in Appendix~\ref{App: {thm: theta=1, xi/2}}.\end{proof}
	
	\begin{remark}
		The results in this section are stronger than  \cite{Z.Jiang-18-1} in three aspects: (i) we gave a simple and explicit expression for the threshold ${\tt T}^*$, while the threshold has to be computed numerically in  \cite{Z.Jiang-18-1}; (ii) we found the asymptotical NAAoI ($\lim_{M\rightarrow\infty} \mathbb{E}[J^{GSAT}(M)]$) analytically; (iii) the threshold in this section can be applied not only to CSMA, but also to any  other transmission policy.
	\end{remark}
	\begin{remark}
		The framework we have built, particularly Algorithm~\ref{alg: General Threshold-slotted ALOHA} and  Theorem~\ref{thm: theta=1, xi/2},  can be directly applied to other settings and multi-access technologies such as MAC with common information in \cite{YODT2015} and a queue-length-based MAC in \cite{JNBTRS2012}. These technologies can achieve sum capacity $1$, like CSMA, and their corresponding normalized age tends to $\frac{1}{2}$ as $M$ gets large.
	\end{remark}

	\section{Numerical Results}\label{sec: Numerical results and discussions}

	\begin{figure*}[ht!]
		\centering
		\begin{subfigure}[b]{0.45\textwidth}
			\input{fig-age1}
			\caption{NAAoI for stationary and adaptive age-based policies.}
			\label{fig-age1}
		\end{subfigure}
		~   \begin{subfigure}[b]{0.45\textwidth}
			\input{fig-age5.tex}
			\caption{Probabilities of successful transmission and  being active in adaptive thinning policies.}
			\label{fig-age5}
		\end{subfigure}
		\caption{NAAoI and success transmission probabilities}
	\end{figure*}

	\begin{figure*}[ht!]
		\centering
		\begin{subfigure}[b]{0.45\textwidth}
			\input{fig-age2_sep1}
			\caption{NAAoI of slotted ALOHA and its comparisons}
			\label{fig-age2_sep1}
		\end{subfigure}
		\begin{subfigure}[b]{0.45\textwidth}
			\input{fig-age2_sep2}
			\caption{NAAoI of CSMA and its comparisons}
			\label{fig-age2_sep2}
		\end{subfigure}
		\caption{NAAoI when $M=500$ v.s $\theta\in(\frac{1}{M},1]$}
		\label{fig-age2}
	\end{figure*}

	\begin{figure*}[ht!]
		\centering
		\begin{subfigure}[b]{0.45\textwidth}
			\input{fig-age6_sep1}
			\caption{NAAoI of slotted ALOHA and its comparison}
			\label{fig-age6_sep1}
		\end{subfigure}
		\begin{subfigure}[b]{0.45\textwidth}
			\input{fig-age6_sep2}
			\caption{NAAoI of Stationary thinning with CSMA}
			\label{fig-age6_sep2}
		\end{subfigure}
		\caption{NAAoI when $M=500$ v.s $\theta\in(0,\frac{1}{M}]$}
		\label{fig-age6}
	\end{figure*}

	In this section, we verify our findings through simulations. 
	Figure~\ref{fig-age1} shows the normalized age under adaptive  and stationary age-based transmission policies for $M=50, 100, 500$. For stationary age-based policies, the normalized age converges to $\frac{e}{2}$ when $M$ is large, validating our findings in Theorem~\ref{thm: theta=1, e/2}. 
	
	The performance of the adaptive policy is better than that of the stationary age-based policy for $\theta>\frac{1}{M}$  and the efficacy (the gap between the two curves) increases with $\theta$. 
	Since the maximum sum throughput of slotted ALOHA is $\frac{1}{e}$, one may ask if this contradicts the lower bound of Proposition~\ref{pro: lowerbound2}. To answer this question, we  remark that the adaptive age-based transmission policy is \emph{not} a slotted ALOHA scheme and therefore the maximum throughput of slotted ALOHA would not apply. As a matter of fact, Fig.~\ref{fig-age5} shows that the throughput of the scheme increases  beyond $\frac{1}{e}$ with $\theta$, supporting Proposition~\ref{pro: lowerbound1}. This is because the AAT policy implicitly facilitates coordination among the transmitters as they utilize the (estimated) age-age distribution for decision making. The throughput at $\theta=1$, as seen in Fig.~\ref{fig-age5}, is close to $.48$ which is consistent with the known lower bound $0.4878$ and upper bound $0.568$ on the (information theoretic) channel capacity of random access channels with feedback \cite{Rga2, Tsybakov}.  It is interesting that the AAT policy can both increase the throughput and decrease the AoI simultaneously when $\theta$ approaches $1$. 
	
	One can also observe  that  the adaptive policy performs worse than the  stationary policy for   $\frac{1}{eM}\leq \theta\leq \frac{1}{M}$. The SAT policy is designed as the stationary version of the AAT policy, thus the SAT policy should not in principle outperform the AAT policy (assuming that our approximations of the estimates of the age-gain distribution are accurate). However, it is worthwhile to discuss this counter-intuitive phenomenon and we expand on the underlying reason: We consistently underestimate the threshold ${\tt T}(k)$ due to the integer nature of it, especially when $\theta$ is small. For example, consider $\frac{1}{eM-1}\leq\theta<\frac{1}{eM-2}$. Note that $\sum_{m\geq1}a_m(k)=\theta$, which is close to $\frac{1}{eM}$. From \eqref{eq: threshold}, in some slots, the event $\{\sum_{m\geq2}\hat{a}_m(k) < \frac{1}{eM}\}$ may occur (even if the estimate $\{\hat{a}_m(k)\}_m$ is exactly accurate), so the threshold ${\tt T}(k)$ under the AAT policy would be $1$ because ${\tt T}(k)$ is always an integer. In these time slots, the AAT policy is reduced to the slotted ALOHA. Therefore, the fraction of active nodes becomes large, the throughput decreases, and the age increases. On the other hand, in the regime when $\frac{1}{eM}\leq\theta\leq\frac{1}{M}$, the threshold ${\tt T}^*=2$ under the SAT policy in every time slot (see Theorem~\ref{thm: limit threshold}). Subsequently, the estimate of the age-gain distribution $\{\hat{\ell}_m(k)\}_m$ is imprecise (see Remark \ref{remarkinaccuracies}). Moreover, the imprecise estimates in \eqref{eq: recursion-F1}  may aggravate the underestimation of the threshold, and the closed-loop  worsens the performance of the AAT policy. The effect is more pronounced in the regime $\frac{1}{eM}\leq\theta\leq\frac{1}{M}$ where old packets are less frequently replaced with new packets.
	
	Finally the age-performance of our proposed distributed age-based policies are compared with the lower bounds of Section~\ref{sec: Lower Bound}, state-of-the-art distributed schemes such as \cite{S.K.Kaul-2017},  as well as centralized Max-Weight policies such as \cite{IKEM2021}. For clarity of exposition, we consider two regimes of $\theta$: $\theta \in (\frac{1}{M}, 1]$ (see Fig.~\ref{fig-age2}), and  $\theta \in (0, \frac{1}{M}]$ (see Fig.~\ref{fig-age6}). Fig.~\ref{fig-age6}, in particular, shows that when $\theta\leq\frac{1}{eM}$,  the normalized age of slotted ALOHA coincides with centralized Max-Weight policies and the lower bound of Proposition~\ref{pro: lowerbound2}. When $\theta$ increases beyond $\frac{1}{eM}$, our proposed age-based thinning methods provide significant gains compared to randomized stationary and slotted ALOHA schemes. When $\theta=\frac{1}{o(M)}$, the NAAoI of slotted ALOHA explodes, and we omit the curve of slotted ALOHA in Fig.~\ref{fig-age2_sep1}. Finally, we  numerically observe that the normalized age of the centralized Max-Weight policy is approximately attained by  stationary age-based thinning in perfect CSMA\footnote{We consider perfect CSMA in simulations. In other words, no errors occur in carrier sensing.} (see the green square curve in Fig.~\ref{fig-age2_sep2} and Fig.~\ref{fig-age6_sep2}),  where the
	length of one contention slot is set to be 1/100 \cite{Z.Jiang-18-1, YGSGBK2011}. 
	
	Next, we compare our proposed algorithms with policies in related works, such as a lazy version of slotted aloha in \cite{DCAEUOK2020},  and variants of CSMA in \cite{AMMAAE2020, sleepwakeconference, AMBYSRSNBS2021}. Different from \cite{DCAEUOK2020}, we considered a random access channel with re-transmission attempts for packets and used a slotted ALOHA with time-variant transmission probability, while in \cite{DCAEUOK2020}, a channel without re-transmission and a slotted ALOHA with a time-invariant transmission probability is investigated. More importantly, we proposed a policy where the best threshold is found in every time slot, while a predetermined threshold is given in \cite{DCAEUOK2020}. Furthermore, we showed the performance analysis for arbitrary $M$ sources under arbitrary generation/arrival rate (in $[0, 1]$), while \cite{DCAEUOK2020} only provided the closed form of average AoI for the case when $M=2$ and $\theta=1$. 
	Compared to the performance of the policy in \cite{DCAEUOK2020}, our proposed AAT and SAT policies outperform the lazy version of slotted aloha (see Fig.~\ref{fig-age2_sep1}, the purple square curve).
	To apply the policies in \cite{AMMAAE2020} on our model, we consider the generation/arrival rate is relatively large ($\theta\geq0.1$) because under  \cite[Assumption~1]{AMMAAE2020}, a transmitter sends a ``fake'' update if its buffer is empty. From Figure~\ref{fig-age2_sep2}, it is easy to see that the stationary thinning with CSMA outperforms the policy in \cite{AMMAAE2020} (Fig~\ref{fig-age2_sep2}, square black curve). 
	References \cite{sleepwakeconference, AMBYSRSNBS2021} are concurrent works on optimizing peak AoI over random access channels with per-source battery lifetime constraints. Translating the introduced energy constraints to arrival rate $\theta$, one can apply the sleep-waking schemes of \cite{sleepwakeconference, AMBYSRSNBS2021} to our problem when CSMA capabilities are available. Using the symmetry of our model along with \cite[Eqns.~(13), (16)]{ AMBYSRSNBS2021}, when $M$ is sufficiently large, the fraction of time in which every source is in transmission mode is around $r/(Mr+1)$ where $r$ is the sleep period parameter as proposed in \cite{sleepwakeconference, AMBYSRSNBS2021}. Using \cite[Eqns.~(13), (16)]{ AMBYSRSNBS2021}, one can argue that we need to consider the so called energy-adequate regime introduced in \cite{sleepwakeconference, AMBYSRSNBS2021} which translates to relatively large $\theta$, i.e., $\theta\geq0.1$. The performance of the policy in \cite{sleepwakeconference, AMBYSRSNBS2021} (Fig~\ref{fig-age2_sep2}, red square point) is similar to that of the stationary thinning with CSMA, which is consistent with the optimality results presented in \cite{sleepwakeconference, AMBYSRSNBS2021}.

	It is worthwhile to mention  that the proposed algorithms not only utilize fully channel capacity, but minimize NAAoI. If we only consider policies with maximum throughput (e.g. standard slotted ALOHA and its variants), the NAAoI explodes up with time for arrival rates above $\frac{1}{e}$. This is also observed in works such as \cite{S.K.Kaul-2017} that adapt slotted ALOHA without packet management for age minimization.

	\section{Conclusion and Future Research}\label{sec: future serach}

	In this work, we investigated the AoI performance of a {\it decentralized} system consisting $M$ source nodes communicating with a common receiver. We first derived a general lower bound on AoI. Then, we derived the analytical (normalized) age performance of (stabilized) slotted ALOHA in the limit of $M\rightarrow\infty$. As the sum arrival rate increases beyond $\frac{1}{e}$, slotted ALOHA becomes unstable. We show that by prioritizing transmissions that offer significant reduction in AoI, we can increase the arrival rate and simultaneously decrease AoI. In particular, we proposed two age-based thinning policies: (i) Adaptive Age-based Thinning (AAT) and (ii) Stationary Age-based Thinning (SAT) and analyzed the age performance in the limit of $M\rightarrow\infty$.  Finally, we demonstrated how our proposed thinning mechanism (SAT) is useful for other random access technologies.  Numerical results showed that the proposed age-based thinning mechanisms make a significantly contribution to the performance of age even for moderate values of $M$.
	Our framework can not be extended to generalized settings (such as \cite{GPAP2003, MFDMDJL2011,  DSJS2012, LJDSJSJW2010}) blindly. Appropriate adaptations related to different settings are necessary. For example, after applying Algorithm~\ref{alg: General Threshold-slotted ALOHA} to transmission schemes in \cite{DSJS2012, MFDMDJL2011}, they are reduced to stationary randomized policies. This is because we do not fully utilize the additional knowledge of the queue length that is provided in these settings. In the setting of \cite{GPAP2003}, we should further assume that the set of nodes that transmit packets are known, and a continuous-time version of the framework is needed in the setting of \cite{LJDSJSJW2010}.

	Future research includes generalization to accommodate 1) dynamic channels, i.e., the number of nodes $M$, or the arrival rates $\theta$ are time-variant 2) asymmetric channels, i.e., the arrival rates $\theta_i$ is different. In the first case, the method we proposed above can be applied directly. Suppose that the expressions of the number of nodes, $M(k)$, and the arrival rates, $\theta(k)$, are known. We cane replace $M$ and $\theta$ by $M(k)$ and $\theta(k)$, respectively, in every time slot. Subsequently, the fixed threshold hold ${\tt T}^*$ is also a time-variant variable, ${\tt T}^*(k)$. In the second case, the method we proposed above can not be applied directly.  This is because we use the profile of all sources as an estimate on any individual source. A more general estimation method should be proposed in the second case. In addition, note that slotted ALOHA algorithms are order optimal when $M$ is sufficiently large and generation/arrival rate is small ($\theta\leq\frac{1}{eM}$). An interesting extension is to consider a smarter decentralized age-based algorithm that can achieve a constant additive age gap from the optimum average age when $\theta$ is small.

	\appendices
	
	\section{Sufficiency of Unit Buffer Size}\label{App: Proof of first lemma}
	
	Consider two types of policies: policies with buffer size $1$, denoted by $\pi_1$, and policies with larger buffer sizes, denoted by $\pi_2$.  To differentiate the two policies and their corresponding queues, we label the packets inside the queues by {\it new} and {\it old}. A new packet in a queue refers to the latest arrival. A packet in a queue is considered old if there is a newer packet in the same queue or if the packet (or a fresher packet) from that source is already delivered at the receiver. In the following, we refer to the freshest old packet as the old packet. At a given time slot, denote the new packet and the old packet of source $i$ by $p^{(i)}_{new}$ and $p^{(i)}_{old}$, respectively. 
	Denote the arrival times of the new and old packets as $t^{(\text{n})}_i$ and $t^{(\text{o})}_i$. It is clear that $t^{(\text{n})}_i>t^{(\text{o})}_i$. 
	We will show that no matter what policy $\pi_2$ does, there is always a policy of type $\pi_1$ whose resulting age is at least as low as $\pi_2$ with respect to every source node.

	At time slot $t'$, suppose policy $\pi_2$ chooses certain action, then we design policy $\pi_1$ to follow the same action with the new packet. In this time slot, under $\pi_2$  a subset of sources  transmit packets. Denote the index of these sources by $\mathcal{I}$. 
	For the sources which do not transmit packets, the AoI under both policies will increase by $1$. For the sources in $\mathcal{I}$, we have the following two cases:
	
	\noindent{\bf Case 1.} Suppose collision happens in time slot $t'$. Then, no packet is delivered, and the AoI of these sources under both policies will increase by $1$.

	\noindent{\bf Case 2.} If a packet is delivered, which implies the cardinal of $\mathcal{I}$, $|\mathcal{I}|=1$. Denote the index of this source by $i$. Then at the next time slot, the AoI under $\pi_1$ drops to $h_i^{\pi_1}(t'+1)=t'-t^{(\text{n})}_i+1$, and the AoI under $\pi_2$ drops to $h_i^{\pi_2}(t'+1)=t'-t^{(\text{o})}_i+1>h_i^{\pi_1}(t'+1)$. This means that from $t'$ onward $h_i^{\pi_2}(t)$ will be point-wise larger or equal to $h_i^{\pi_1}(t)$, $t>t'$.

	\section{Proof of Proposition~\ref{pro: lowerbound1}.}\label{App: lowerbound1} 
	Consider any transmission policy and a large time-horizon $K$. Let $L_i$ be the number of remaining time slots after the last packet delivery in source $i$. The NAAoI defined in \eqref{eq: EWSAoI-1} can be re-written as follows:
	\begin{align}
		J_K^{\pi}&=\frac{1}{M^2}\sum_{i=1}^{M}\frac{1}{K}\sum_{k=1}^{K}h_i(k)\nonumber\\
		&=\frac{1}{M^2}\sum_{i=1}^{M}\frac{1}{K}\Big(\sum_{m=1}^{N_i(K)}\Gamma_i(m)+\frac{1}{2}L_i^2\nonumber\\
		&+D_i(N_i(K))L_i-\frac{1}{2}L_i\Big),
	\end{align} 
	where $\Gamma_i(m)$ was expressed in \eqref{eq:gamma}.
	Since $D_i(m)\geq1$ for all $1\leq m\leq N_i(K)$, we can lower bound $\Gamma_i(m)$ by substituting $D_i(m-1)=1$ in \eqref{eq:gamma}. Using similar steps as \cite[Eqns. (9) - (14)]{I.Kadota-2018}, we find
	\begin{align}
		\label{eq:intermediatestep}
		J^{\pi}(M)\geq \lim_{K\to\infty}\mathbb{E}[\frac{1}{2M^2}\sum_{i=1}^{M}\frac{K}{N_i(K)}+\frac{1}{2M}].
	\end{align}
	Recall that $N_i(K)$ is the total number of packets delivered by source $i$. In the limit of $K\rightarrow\infty$, $\frac{N_i(K)}{K}$ is the throughput of source $i$. By the model assumption, in every time slot, at most one packet is delivered in the system. Therefore, 
	\begin{align}
		\lim_{K\rightarrow\infty}\mathbb{E}[\sum_{i=1}^{M}\frac{N_i(K)}{K}]\leq C_{RA}.\label{eq:CRA}
	\end{align}
	Now note that by the Cauchy-Schwarz inequality, we have
	\begin{align}
		\lim_{K\rightarrow\infty}\mathbb{E}[\sum_{i=1}^{M}\frac{N_i(K)}{K}]\mathbb{E}[\sum_{i=1}^{M}\frac{K}{N_i(K)}]\geq M^2.\label{eq:M2}
	\end{align}
	Thus using \eqref{eq:CRA} and \eqref{eq:M2}, we find
	\begin{align}
		\lim_{K\rightarrow\infty}\mathbb{E}[\sum_{i=1}^{M}\frac{K}{N_i(K)}]\geq\frac{M^2}{C_{RA}}.
		\label{eq:CS}
	\end{align}
	Inserting \eqref{eq:CS} back into \eqref{eq:intermediatestep}, we obtain
	\begin{align}
		J^{\pi}_M\geq \frac{1}{2C_{RA}}+\frac{1}{2M}.
	\end{align}

	\section{Proof of Proposition~\ref{pro: lowerbound2}.}\label{App: lowerbound2} 
	Suppose all packets are delivered instantaneously with one time-unit delay and without experiencing collisions. A lower bound to NAAoI in this scenario constitutes a lower bound to NAAoI in our setup. Let $X_i(m)$ denote the inter arrival time between the $m$th and $(m+1)$st packets. $\{X_i(m)\}_m$ is  a geometric iid sequence. Under the assumption of instantaneous delivery, $I_i(m)=X_i(m)$. It hence follows from \eqref{eq:gamma}  that
	\begin{align}
		\Gamma_i(m)=\sum_{k=T_i(m)}^{T_i(m)+X_i(m)-1}h_i(k)=\frac{1}{2}X_i(m)^2+\frac{1}{2}X_i(m).
	\end{align}
	Thus, similar with \cite{I.Kadota-2018}, the time-average AoI of source $i$, denoted by $H_i$, is
	\begin{align}
		\mathbb{E}[H_i]=\lim_{K\rightarrow\infty}\frac{1}{K}\sum_{k=1}^{K}h_i(k)=\frac{\mathbb{E}[X^2]}{2\mathbb{E}[X]}+\frac{1}{2}.\label{eq:GRV}
	\end{align}
	Since $X$ in \eqref{eq:GRV} has a geometric distribution with parameter $\theta$, we find 
	\begin{align}
		\mathbb{E}[H_i]=\frac{2-\theta}{2\theta}+\frac{1}{2}.
	\end{align}
	Note that NAAoI can be captured by 
	\begin{align*}
		\frac{1}{M^2}\sum_{i=1}^{M}H_i
	\end{align*}
	and one can hence conclude that
	\begin{align}
		J^{\pi}(M)\geq\frac{1}{M\theta}.
	\end{align}

	\section{Proof of Proposition~\ref{thm: MW}.}\label{APP: Proof of MW}
	
	First consider a source node $i$ whose queue is empty. This means that no new packet has arrived at that transmitter since the last delivery (from that source node) at the receiver; i.e., $h_i(k)=w_i(k)$ and hence $\delta_i(k)=0$. Such nodes $i$ are thus irrelevant because $\delta_j(k)\geq 0$ for all source nodes $j$. Now consider nodes with non-empty queues. Among these nodes, $d_i(k)$ is non-zero if and only if $(\lambda_1,\ldots,\lambda_M)$ is a vector consisting of $0$'s except for $\lambda_i=1$. Hence at most one $d_i(k)$ can be equal to $1$. Call the corresponding source node $\ell(k)$.
	Expression \eqref{eq:drift} is minimized when $d_{\ell(k)}(k)$ picks the largest $\delta_j(k)$.

	\section{Probabilities of idle, deliveries and collisions under Slotted ALOHA}\label{App: {lem: slotted aloha probabilities}} 
	
	\begin{lemma}\label{lem: slotted aloha probabilities}
		Consider any stabilized slotted ALOHA scheme. Define $G$ as the expected number of attempted transmissions in a slot. Then, for $M$ large, the probability of delivering a packet is (asymptotically) $G e^{-G}$, the probability of idle system is (asymptotically) $e^{-G}$, and the probability of collisions is (asymptotically) $1-e^{-G}-Ge^{-G}$. In particular, when $G=1$, the maximum probability of delivery is $1/e$, the corresponding probabilities of collisions and idle system are $1-2/e$ and $1/e$, respectively. 
	\end{lemma}
	\begin{proof}
		The idea of the proof is very similar to \cite[Chapter~4]{B-D.Bertsekas}. However, the settings are different: \cite[Chapter~4]{B-D.Bertsekas} considered that  packets arrive as a Poisson process (in a continuous-time system) and the buffer size is infinite, while this proof consider that packets arrive as a Bernoulli process (in a discrete-time system) and the buffer size is $1$. Define the nodes that are not backlogged as {\it fresh} nodes. Each fresh node transmits a packet directly in a slot if it is not empty, and it generates/receives a packet with probability $\theta$, thus a fresh node transmits a packet with probability $\theta$.  
		Let $P_a\big(i, n(k)\big)$ be the probability that $i$ fresh nodes transmit a packet in a time slot and let $P_{s}\big(j, n(k)\big)$ be the probability that $j$ backlogged nodes transmit. We have:
		\begin{align}
			P_a\big(i, n(k)\big)=&\binom{M-n(k)}{i}(1-\theta)^{M-n(k)-i}\theta^{i}\label{eq: P_a}\\
			P_s\big(j, n(k)\big)=&\binom{n(k)}{j}(1-p_b(k))^{n(k)-i}p_b(k)^{i}\label{eq: P_s}.
		\end{align}
		Thus, in slot $k$, when a packet is delivered, i.e., $\sum_{i=1}^{M}d_i(k)=1$, the probability is
		\begin{equation}\label{eq: P_delivery}
			\begin{aligned}
				&\Pr(\sum_{i=1}^{M}d_i(k)=1)\\
				=&P_a\big(1,n(k)\big)P_s\big(0,n(k)\big)+P_a\big(0,n(k)\big)P_s\big(1, n(k)\big).
			\end{aligned}
		\end{equation}
		If the channel does not transmit a packet in a slot, i.e., we have an idle channel, $\sum_{i=1}^{M}d_i(k)=0$, $c(k)=0$. The probability of idle system in slot $k$ is
		\begin{align}\label{eq: P_idle}
			\Pr(\sum_{i=1}^{M}d_i(k)=0,c(k)=0)=P_a\big(0,n(k)\big)P_s\big(0,n(k)\big).
		\end{align}
		Define the attempt rate $G=(M-n(k))\theta+n(k)p_b(k)$ as the expected number of attempted transmissions in a slot. 
		From \eqref{eq: P_a} and \eqref{eq: P_s}, the probability of delivery is
		\begin{align*}
			&\Pr(\sum_{i=1}^{M}d_i(k)=1)\\
			=&\big(M-n(k)\big)(1-\theta)^{M-n(k)-1}\theta (1-p_b)^{n(k)}\\
			+&(1-\theta)^{M-n(k)}n(k)(1-p_b)^{n(k)-1}p_b
		\end{align*}
		and the probability of an idle channel is
		\begin{align*}
			\Pr(\sum_{i=1}^{M}d_i(k)=0, c(k)=0)=(1-\theta)^{M-n(k)}(1-p_b)^{n(k)}.
		\end{align*}
		Note that the valid regime of $\theta$ is $\theta M<\frac{1}{e}$, and thus $\theta$, $p_b$ are small. Using the approximation $(1-x)^{-y}\approx\exp(-x y)$ for small $x$, we find 
		\begin{align*}
			&\Pr(\sum_{i=1}^{M}d_i(k)=1)\approx G e^{-G} \\
			&\Pr(\sum_{i=1}^{M}d_i(k)=0,c(k)=0)\approx e^{-G}.
		\end{align*}
		\begin{align*}
			\Pr(c(k)=1)\approx 1-G e^{-G}-e^{-G}.
		\end{align*}
		Taking the first derivative of the function $Ge^{-G}$, we can find the maximum point is $1$ for $0<G\leq1$. So the maximum probability of delivery is 
		\begin{align}
			\Pr(\sum_{i=1}^{M}d_i(k)=1)\approx1/e,
		\end{align} 
		correspondingly, we have
		\begin{align}
			&\Pr(\sum_{i=1}^{M}d_i(k)=0, c(k)=0)\approx 1/e,\\
			&\Pr(c(k)=1)=1-2/e.
		\end{align}
	\end{proof}

	\section{Proof of Theorem~\ref{thm: slotted aloha age e}.}\label{App: {thm: slotted aloha age e}}
	
	The proof is organized in three parts:
	
	\noindent{ Part 1}: Preliminaries.
	In time slot $k$,  denote the time just before the arrival of new packets by $k^-$ and the time just after the arrival of new packets by $k^+$. We hence write $\delta_i(k^-)=h_i(k^-)-w_i(k^-)$ and $\delta_i(k^+)=h_i(k^+)-w_i(k^+)$. 
	Suppose a packet is delivered  from the $i^{th}$ source at the end of time slot $k-1$. We then have $\delta_i\big(k^{-}\big)=0$. From \eqref{eq: aloha}, since all  nodes have the same arrival rate and transmission policy,  the sequences $\{h_i(k^-)\}_{k=1}^{\infty}$, $\{h_i(k^+)\}_{k=1}^{\infty}$, $\{w_i(k^-)\}_{k=1}^{\infty}$, $\{w_i(k^+)\}_{k=1}^{\infty}$, $\{\delta_i(k^-)\}_{k=1}^{\infty}$, $\{\delta_i(k^+)\}_{k=1}^{\infty}$ are identical random variables across $i=1,2,\cdots, M$, respectively. 
	Recall that source nodes with $\delta_i(k^-)=0$ are  $0$-order nodes and define $n_0(k^-)$ as the number of $0$-order nodes at time $k^-$. 
	
	In the beginning of time slot $k$, on average,  $\theta M$ new packets arrive at the sources, and $\theta n_0(k^-)$ $0$-order nodes receive new packets. 
	Suppose  source $i$ is a $0$-order node and $h_i(k^-)-w_i(k^-)=0$.  If source $i$ receives new packets, then the source's AoI changes from $w_i(k^-)$ to $w_i(k^+)=0$ and the destination's AoI $h_i(k^-)$ remains the same as $h_i(k^+)$. Thus, 
	\begin{align*}
		\delta_i(k^+)=&h_i(k^+)-w_i(k^+)\\=&h_i(h^+)\\=&h_i(k^-)\\
		>&h_i(k^-)-w_i(k^-)\\=&\delta_i(k^-)\\=&0,
	\end{align*}
	which implies that if a $0$-order source receives a new packet, then it is not a $0$-order source at $k^+$. 
	
	Fix any large $M$ and denote the maximum throughput of Slotted ALOHA with $C_{SA}(M)$. We know that
	\begin{align*}
		\lim_{M\rightarrow\infty}C_{SA}(M)=e^{-1}.
	\end{align*}
	The recursion of the expected number of $0-$order nodes is:
	\begin{equation}\label{eq: proof n0}
		\begin{aligned}
			&\mathbb{E}[n_0\big((k+1)^-\big)]\\
			&=(1-\theta)\mathbb{E}[n_0(k^-)]+\min(M\theta,C_{SA}(M))
		\end{aligned}
	\end{equation}
	where the second term on the right hand side is the average number of delivered packets per time slot. 
	Since we consider a stabilized slotted ALOHA,  $\lim_{k\rightarrow\infty}\mathbb{E}[n_0(k^-)]$ exists. Denote $$n_0^*=\lim_{k\rightarrow\infty}\mathbb{E}[n_0(k^-)].$$
	Letting $k\to\infty$ on both sides of \eqref{eq: proof n0}, we have
	\begin{align}\label{eq: proof n0_1}
		n_0^* = (1-\theta)n_0^* + \min(M\theta,C_{SA}(M)).
	\end{align}
	Note that 
	\begin{align}\label{eq: proof n0_2}
		\lim_{M\rightarrow\infty}\min(M\theta,C_{SA}(M))=\lim_{M\rightarrow\infty} M\theta=\eta.
	\end{align}
	From \eqref{eq: proof n0_1} and \eqref{eq: proof n0_2}, we have
	\begin{align}
		\lim_{M\rightarrow\infty}\frac{n_0^*}{M}=1.\label{app-goestoone}
	\end{align}

	\noindent{ Part~2:} Find the expression of NAAoI. Using \eqref{eq: EWSAoI-1},  we have
	\begin{align*}
		J^{SA}(M)=\lim_{K\rightarrow\infty}\mathbb{E}\left[\frac{1}{M^2}\sum_{i=1}^{M}\frac{1}{K}\sum_{k=1}^{K}h_i(k^-)\right]\triangleq J_1+J_2
	\end{align*}
	where
	\begin{align*}
		J_1=&\lim_{K\rightarrow\infty}\mathbb{E}\left[\frac{1}{M^2}\sum_{i=1}^{M}\frac{1}{K}\sum_{k=1}^{K}w_i(k^-)\right]\\
		J_2=&\lim_{K\rightarrow\infty}\mathbb{E}\left[\frac{1}{M^2}\sum_{i=1}^{M}\frac{1}{K}\sum_{k=1}^{K}\delta_i(k^-)\right].
	\end{align*}

	\noindent{ Part 3}: Find the limit of NAAoI. First, we consider $J_1$. $w_i(k^-)$ has a geometric distribution starting from $1$ with parameter $\theta$ for all $i$. Employing the law of large number, we find
	\begin{align}
		J_1=\frac{1}{M\theta}.
	\end{align}

	Next, we consider $J_2$ and prove that its limit in large $M$ approaches zero.
	Note from \eqref{eq: delta} that $\delta_i(k)=0$ if source $i$ is empty in time slot $k$ and $\delta_i(k)>0$ if a packet remains in source $i$ in time slot $k$. We first note that $\delta_i(k)$ is upper bounded by $h_i(k)$. Let us consider a worse case in which buffer sizes are infinite. In this case, assuming stationarity\footnote{This assumption approximately holds for infinite time horizon $T$}, denote the inter-arrival time and delay of packets with respect to source $i$ by $X_i$ and $D_i$.  Since the Bernoulli arrival process has parameter $\theta$, we have $\mathbb{E}[X_i]=\frac{1}{\theta}=\frac{M}{\eta}$. Moreover $\mathbb{E}[D_i]$ is approximately bounded by some constant independent of the number of sources $M$ \cite{ModinaE2009}. Now we observe that for each packet delivery the expected peak age at the destination is upper bounded by $\mathbb{E}[X_i]+\mathbb{E}[D_i]$. We can hence write
	\begin{align}\label{eq: Edelta}
		\mathbb{E}[\delta_i(k)|\delta_i(k)>0]\leq\mathbb{E}[X_i]+\mathbb{E}[D_i]
	\end{align}
	which implies that $\mathbb{E}[\delta_i(k)]$ is $O(M)$.

	Now expand $J_2$:
	\begin{align}
		J_2=&\lim_{K\rightarrow\infty}\mathbb{E}[\frac{1}{M^2}\sum_{i=1}^{M}\frac{1}{K}\sum_{k=1}^{K}\delta_i(k^-)]\nonumber\\
		=&\lim_{K\rightarrow\infty}\mathbb{E}[\frac{1}{M^2}\frac{1}{K}\sum_{k=1}^{K}\sum_{i=1}^M\delta_i(k^-)1_{ \delta_i(k^-)>0}]\nonumber\\
		=&\lim_{K\rightarrow\infty}\frac{1}{M^2}\frac{1}{K}\sum_{k=1}^{K}\sum_{i=1}^M\mathbb{E}[\delta_i(k^-)1_{ \delta_i(k^-)>0}]\nonumber\\
		\leq&\limsup_{k\rightarrow\infty}\frac{1}{M^2}\sum_{i=1}^M\mathbb{E}[\delta_i(k^-)1_{ \delta_i(k^-)>0}]\nonumber\\
		=&\limsup_{k\rightarrow\infty}\frac{1}{M^2}\sum_{i=1}^M\Big(\Pr(\delta_i(k^-)>0)\nonumber\\
		&\times\mathbb{E}[\delta_i(k^-)|\delta_i(k^-)>0]\Big).
	\end{align}
	Since for $k$ large enough the conditional expectation $\mathbb{E}[\delta_i(k^-)|\delta_i(k^-)>0]$ is $O(M)$, it remains to prove that in the limit of large M, $\lim_{k\to\infty}\frac{1}{M}\sum_{i=1}^M\Pr(\delta_i(k^-)>0)$ vanishes. But this holds because we can write
	\begin{align}
		&\lim_{k\to\infty}\frac{1}{M}\sum_{i=1}^M\Pr(\delta_i(k^-)>0)\nonumber\\
		&=\lim_{k\to\infty}\mathbb{E}[\frac{1}{M}\sum_{i=1}^M1_{ \delta_i(k^-)>0}]\nonumber\\
		&=\lim_{k\to\infty}\mathbb{E}[\frac{1}{M}(M-n_0(k^-))]\nonumber\\
		&=\frac{M-n_0^*}{M}\label{app-goestozero}
	\end{align}
	and \eqref{app-goestozero} goes to zero by \eqref{app-goestoone}.

	Finally, we  prove that for any scheme, $J^{SA}(M)$ is lower bounded by $1/\eta$. From Proposition~\ref{pro: lowerbound2} (and letting  $M\rightarrow\infty$), we have \begin{align*}
		\lim_{M\rightarrow\infty}J^{SA}(M)\geq\lim_{M\rightarrow\infty}\frac{1}{M\theta}=\frac{1}{\eta}.
	\end{align*}
	Therefore, slotted ALOHA can reach the lower bound when $\theta\in(0,\frac{1}{eM}]$ and is hence optimal.

	\section{ Proof of Lemma~\ref{lem: ell_plus}.}\label{App: {lem: ell_plus}}
	Before presenting the proof, we state the following straightforward lemma (whose proof is omitted).
	\begin{lemma}\label{lem: TSA w_i(k+1)>0}
		At the beginning of time slot $k$, before new packets arrive at source $i$, $w_i(k^-)>0$ and its probability distribution is
		\begin{equation}\label{eq: conditional distribution of w}
			\begin{aligned}
				\Pr\big(w_i(k^-)=j\big)=\theta(1-\theta)^{j-1},\,\, j=1,2,3,\cdots.
			\end{aligned}
		\end{equation}
	\end{lemma}

	First consider $m=0$ and suppose source $i$ is a $0$-order node. From Lemma~\ref{lem: TSA w_i(k+1)>0}, we know that $w_i(k^-)>0$. Moreover, since $\delta_i(k^-)=0$, we conclude $h_i(k^-)=w_i(k^-)>0$. Once the $0$-order node has a new arrival, $w_i(k^+)=0$ and $h_i(k^+)=h_i(k^-)$, resulting in $\delta_i(k^+)=h_i(k^+)>0$; i.e., the order of the node increases. In other words, the order of a $0$-order node increases once it receives a new packet. In total, the fraction of $0$-order nodes that become of higher order is on average $\theta \ell_0(k-1)$. Thus,
	\begin{align*}
		\ell_0(k^+)=(1-\theta)\ell_0(k-1).
	\end{align*}
	Similarly, we consider $m\geq1$. The fraction of $m$-order nodes that have new arrivals is $\theta \ell_m(k-1)$. These nodes will have larger orders. 
	Suppose source $i$ is of order $m$, $m\geq1$, i.e., $\delta_i(k^-)=h_i(k^-)-w_i(k^-)=m$, once a new packet arrives, then $w_i(k^+)=0$, $h_i(k^+)=h_i(k^-)$, and $\delta_i(k^+)=h_i(k^+)=m+w_i(k^-)$. From Lemma~\ref{lem: TSA w_i(k+1)>0}, $w_i(k^-)>0$, then $\delta_i(k^+)>\delta_i(k^-)=m$. The order of a $m$-order node increases once it receives a new packet. In total, the fraction of $m$-order nodes have larger orders is $\theta\ell_m(k-1)$.

	More precisely, consider a $j$-order node, $j<m$. This node becomes an $m$-order node if it receives a new packet and $w_i(k^{-})=m-j$. Using Lemma~\ref{lem: TSA w_i(k+1)>0}, we cam write
	\begin{align}
		\ell_m(k^+)=&(1-\theta)\ell_m(k-1)\nonumber\\
		+&\sum_{j=0}^{m-1}\theta \ell_j(k-1) \Pr\Big(w_i(k^-)=m-j\Big)\nonumber\\
		=&(1-\theta)\ell_m(k-1)+\sum_{j=0}^{m-1}\theta \ell_j(k-1) \theta(1-\theta)^{m-j-1}\nonumber\\
		=&(1-\theta)\ell_m(k-1)+\theta^2\sum_{j=0}^{m-1}\ell_j(k-1) (1-\theta)^{m-j-1}\label{eq: ell_m_plus}.
	\end{align}
	where the second term in \eqref{eq: ell_m_plus} on the left hand side is the average fraction of nodes that have just become of order $m$. Denoting it by $a_m$, we have\begin{align}
		a_m(k)=\theta^2\sum_{j=0}^{m-1}\ell_j(k-1) (1-\theta)^{m-j-1}.
	\end{align}

	\section{ Proof of Lemma~\ref{thm: N', E_m}.}\label{App: {thm: N', E_m}}
	From the expression of $\ell^*_m$ in \eqref{eq: limit collection}, $0\leq m\leq T^*-1$, we obtain 
	\begin{align}
		\ell_0^*&=\frac{1}{eM\theta}\\
		\ell_m^*&=\frac{a^*_m}{\theta}\qquad 0\leq m\leq T^*-1\label{app-a-ell}.
	\end{align} 
	From \eqref{def-astar}, $a_m^*$ depends on $\{\ell_j^*\}_{j\leq m-1}$ and from \eqref{app-a-ell}, $\ell^*_m$ depends on $a_m^*$ for $1\leq m\leq T^*-1$. So they can be recursively found and in particular, it is not difficult to prove  for all $1\leq m\leq T^*-1$:
	\begin{align}
		a_m^*&=\frac{\theta}{eM}\label{app-a}\\
		\ell_m^*&=\frac{1}{eM}.\label{app-ell}
	\end{align} 
	We prove this by mathematical induction on $T^*\geq 2$. For $T^*-1=1$, the statement holds because $$a_1^*=\theta^2 \ell_0^*=\frac{\theta}{eM}$$$$\ell_1^*=\frac{a_1^*}{\theta}=\frac{1}{eM}.$$
	
	Now suppose the statements \eqref{app-a}-\eqref{app-ell} hold for $m\leq {\tt T}^*-1=k$. We prove the statement for ${\tt T}^*-1=k+1$ and in particular we find $a_{k+1}^*$ and $\ell_{k+1}^*$ below:
	\begin{align}
		&a_{k+1}^*=\theta^2\sum_{j=0}^{k}\ell_j^* (1-\theta)^{k-j}\nonumber\\
		=&\theta^2\frac{1}{eM}\sum_{j=1}^{k}(1-\theta)^{k-j}+\theta^2(1-\theta)^k\frac{1}{eM\theta}\nonumber\\
		=&\theta^2\frac{1}{eM}\frac{1-(1-\theta)^{k}}{\theta}+\theta(1-\theta)^k\frac{1}{eM}\nonumber\\
		=&\frac{\theta}{eM}.\label{app-arecursion}
	\end{align}
	Next, using  \eqref{app-a-ell}, we find \begin{align}
		\ell_{k+1}=\frac{1}{eM}.
	\end{align}
	Moreover, using the derivation in \eqref{app-arecursion}, we also find
	\begin{align}
		a_{T^*}^*=&\frac{\theta}{eM}.
	\end{align}
	Finally, from \eqref{eq: limit-threshold1}, we obtain \begin{align}\ell_m^{+*}=\frac{1}{eM}\qquad 1\leq m\leq {\tt T}^*-1.\end{align}

	\section{Proof of Theorem~\ref{thm: limit threshold}.}\label{App: {thm: limit threshold}}
	Summing \eqref{eq: limit collection} on both sides, we have
	\begin{align*}
		\sum_{m\geq1}a_m^*=\theta.
	\end{align*}
	Moreover,  ${\tt T}^*$ satisfies
	\begin{align}\label{eq: Appendix_Thresh}
		{\tt T}^*=\max\{t|\sum_{m\geq t}a_m^*\geq \frac{1}{eM}\}
	\end{align}
	by its definition in \eqref{eq: threshold}. The term $\sum_{m\geq {\tt T}^*}a_m^*$ can be re-written as follows:
	\begin{align}
		\sum_{m\geq {\tt T}^*}a_m^*=&\sum_{m\geq 1}a_m^*-\sum_{m< {\tt T}^*}a_m^*\nonumber\\\stackrel{(a)}{=}&\theta-({\tt T}^*-1)\frac{\theta}{eM}\label{eqsumem1}
	\end{align}
	where $(a)$ follows by \eqref{eq: a_mstar} in Lemma~\ref{thm: N', E_m}. On the other hand, $\sum_{m\geq {\tt T}^*}a_m^*$ satisfies the following inequality by  \eqref{eq: Appendix_Thresh}:
	\begin{align}
		\sum_{m\geq {\tt T}^*}a_m^*\geq\frac{1}{eM}. \label{eqsumem}
	\end{align}
	Putting \eqref{eqsumem1} and \eqref{eqsumem} together, we find
	\begin{align}
		{\tt T}^*=\lfloor eM-\frac{1}{\theta}+1 \rfloor
	\end{align}
	since  ${\tt T^*}$ is an integer.

	\section{Proof of Theorem~\ref{thm: theta=1, e/2}.}\label{App: proof of Theorem thm: theta=1, e/2}
	
	The proof is organized in three parts:

	\noindent{ Part 1:} Preliminaries.
	In this part, we discuss some notations and preliminaries which will be used in the proof.
	Denote the time just before arrival of new packets by $k^-$ and the time just after arrival of new packets by $k^+$. Since we have assumed that all nodes are identical,  the sequence  $\{h_i(k^+)\}_{k=1}^{\infty}$ is identical (but not  independent) across all $i=1,2,\cdots, M$. From \eqref{eq: w_i},  $\{w_i(k^+)\}_{k=1}^{\infty}$ are i.i.d with respect to $i$. Therefore, the sequence  $\{\delta_i(k^+)\}_{k=1}^{\infty}$ is identical but not independent for all $i=1,2,\cdots, M$.

	Since $\theta=\frac{1}{o(M)}$ and in particular $\theta>\frac{1}{eM}$,  from Lemma~\ref{thm: N', E_m}, $\ell_m^{+*}=\frac{1}{eM}$ for $m=1,2,\cdots, {\tt T}^*$ and $\ell_0^{+*}=\frac{o(M)}{eM}$. From Theorem~\ref{thm: limit threshold}, ${\tt T}^*=\lfloor eM-1/\theta+1\rfloor=\lfloor eM-o(M)+1\rfloor$. Denote $$s_{{\tt T}^*}=\sum_{m=0}^{{\tt T}^*-1}\ell_m^{+*}.$$ In the limit of large $M$, we have
	\begin{align}\label{eq: lim_s_star}
		\lim_{M\rightarrow\infty}s_{{\tt T}^*}=\lim_{M\rightarrow\infty}\frac{o(M)+\lfloor eM-o(M)+1\rfloor-1}{eM}=1.
	\end{align} 
	The expected number of inactive nodes is $Ms_{{\tt T}^*}$ 
	and the expected number of active nodes is $M(1-s_{{\tt T}^*})$.

	\noindent{ Part 2}: Find the expression of NAAoI.
	Let $\alpha_i=\frac{1}{M}$ for $i=1,2,\cdots, M$ in  \eqref{eq: EWSAoI-1}:
	\begin{align*}
		J^{SAT}(M)=\lim_{K\rightarrow\infty}\mathbb{E}[\frac{1}{M^2}\sum_{i=1}^{M}\frac{1}{K}\sum_{k=1}^{K}h_i(k^+)]\triangleq J_1+J_2
	\end{align*}
	where
	\begin{align*}
		J_1=&\lim_{K\rightarrow\infty}\mathbb{E}[\frac{1}{M^2}\sum_{i=1}^{M}\frac{1}{K}\sum_{k=1}^{K}w_i(k^+)]\\
		J_2=&\lim_{K\rightarrow\infty}\mathbb{E}[\frac{1}{M^2}\sum_{i=1}^{M}\frac{1}{K}\sum_{k=1}^{K}\delta_i(k^+)].
	\end{align*}
	In addition, $J_2=J_{21}+J_{22}$,
	where
	\begin{align*}
		J_{21}=&\lim_{K\rightarrow\infty}\mathbb{E}[\frac{1}{K}\sum_{k=1}^{K}\frac{1}{M^2}\sum_{i: \delta_i(k^+)<{\tt T}^*}\delta_i(k^+)]\\
		J_{22}=&\lim_{K\rightarrow\infty}\mathbb{E}[\frac{1}{K}\sum_{k=1}^{K}\frac{1}{M^2}\sum_{i: \delta_i(k^+)\geq{\tt T}^*}\delta_i(k^+)].
	\end{align*}

	\noindent{ Part 3}: Find the limit of NAAoI with respect to $M$.
	First, we consider $J_1$. From \eqref{eq: w_i}, $w_i(k^+)$ has a geometric distribution with parameter $\theta$ \big(with $w_i(k^+)=0,1,2,\cdots$\big) for all $i$. Let $w$ have the same distribution as $w_i(k^+)$. We thus have
	\begin{align}\label{eq: eqn_J_1}
		J_1=\frac{1}{M}\mathbb{E}[w]=\frac{1-\theta}{M\theta}.
	\end{align} 
	Next, we consider $J_{21}$:
	\begin{align}
		&\lim_{M\rightarrow\infty}J_{21}\nonumber\\
		=&\lim_{M\rightarrow\infty}\lim_{K\rightarrow\infty}\mathbb{E}[\frac{1}{K}\sum_{k=1}^{K}\frac{1}{M^2}\sum_{i: \delta_i(k^+)<{\tt T}^*}\delta_i(k^+)]\nonumber\\
		=&\lim_{M\rightarrow\infty}\lim_{K\rightarrow\infty}\mathbb{E}[\frac{1}{K}\sum_{k=1}^{K}\frac{1}{M^2}\sum_{i=1}^{M}\delta_i(k^+)1_{(\delta_i(k^+)<{\tt T}^*)}]\nonumber\\
		=&\lim_{M\rightarrow\infty}\lim_{K\rightarrow\infty}\mathbb{E}[\frac{1}{K}\sum_{k=1}^{K}\frac{1}{M}\sum_{j=1}^{{\tt T}^*-1}\frac{\sum_{i=1}^{M}\delta_i(k^+)1_{(\delta_i(k^+)=j)}}{M}]\nonumber\\
		=&\lim_{M\rightarrow\infty}\lim_{K\rightarrow\infty}\mathbb{E}[\frac{1}{K}\sum_{k=1}^{K}\frac{1}{M}\sum_{j=1}^{{\tt T}^*-1}\frac{j\sum_{i=1}^{M}1_{(\delta_i(k^+)=j)}}{M}]\nonumber\\
		=&\lim_{M\rightarrow\infty}\lim_{K\rightarrow\infty}\frac{1}{K}\sum_{k=1}^{K}\frac{1}{M}\sum_{j=1}^{{\tt T}^*-1}j\frac{\mathbb{E}[\sum_{i=1}^{M}1_{(\delta_i(k^+)=j)}]}{M}.
	\end{align}
	Substituting $\ell_j(k^+)$ for the term $\frac{\mathbb{E}[\sum_{i=1}^{M}1_{(\delta_i(k^+)=j)}]}{M}$, we find
	\begin{align}
		\lim_{M\rightarrow\infty}J_{21}=\lim_{M\rightarrow\infty}\lim_{K\rightarrow\infty}\frac{1}{K}\sum_{k=1}^{K}\frac{1}{M}\sum_{j=1}^{{\tt T}^*-1}j\ell_j(k^+).
	\end{align}
	By stationarity, note that
	\begin{align*}
		\ell_j^{*+}=\lim_{k\rightarrow\infty}\ell_j(k^+).
	\end{align*}
	By the Cesaro Mean Lemma, 
	\begin{align*}
		\lim_{K\rightarrow\infty}\frac{\sum_{k=1}^{K}\ell_j(k^+)}{K}=\ell_j^{*+}.
	\end{align*}
	Therefore,
	\begin{align}
		&\lim_{M\rightarrow\infty}J_{21}\nonumber\\&=\lim_{M\rightarrow\infty}\frac{1}{M}\sum_{j=1}^{{\tt T}^*-1}j \ell_j^{*+}\nonumber\\
		&=\lim_{M\rightarrow\infty}\frac{1}{M}\frac{\tt T^*(\tt T^*-1)}{2} \frac{1}{eM}=\frac{e}{2}
	\end{align}
	where in the last step we have substituted $\ell_j^{*+}=\frac{1}{eM}$ for $j=1,\ldots,{\tt T}^*-1$ (see Lemma \ref{thm: N', E_m}).
	
	Finally, we consider $J_{22}$:  
	\begin{align}
		&\lim_{M\rightarrow\infty}J_{22}\nonumber\\
		=&\lim_{M\rightarrow\infty}\lim_{K\rightarrow\infty}\frac{1}{K}\sum_{k=1}^{K}\frac{1}{M^2}\sum_{i=1}^M\mathbb{E}[\delta_i(k^+)1_{\delta_i(k^+)\geq{\tt T}^*}]\nonumber\\
		=&\lim_{M\rightarrow\infty}\lim_{K\rightarrow\infty}\frac{1}{K}\sum_{k=1}^{K}\frac{1}{M^2}\sum_{i=1}^M\Big(\mathbb{E}[\delta_i(k^+)|\delta_i(k^+)\geq{\tt T}^*]\nonumber\\
		&\hspace{5cm}\times\Pr(\delta_i(k^+)\geq{\tt T}^*)\Big)\nonumber\\
		\stackrel{(a)}{\leq}&\lim_{M\rightarrow\infty}\lim_{K\rightarrow\infty}\frac{1}{K}\sum_{k=1}^{K}\frac{1}{M^2}\sum_{i=1}^McM\Pr(\delta_i(k^+)\geq{\tt T}^*).
	\end{align}
	In the above chain of inequalities, step $(a)$ holds because $\mathbb{E}[\delta_i(k^+)|\delta_i(k^+)\geq{\tt T}^*]=O(M)$. To show this, we first observe that $\delta_i(k)$ is increasing in $k$ until a delivery occurs. Now, note that $\delta_i(k^+)$ is upper bounded by ${\tt T}^*$ plus the peak age at the first delivery after time slot $k$. The peak age is  bounded by $X_i$ (the inter arrival time), which is $o(M)$ on average, plus delay $D_i$, which is constant on average (similar to \eqref{eq: Edelta}). The threshold ${\tt T}^*$ is also $O(M)$. So overall, we have
	$$\mathbb{E}[\delta_i(k^+)|\delta_i(k^+)\geq{\tt T}^*]\leq cM$$
	for some constant $c$. Note that
	\begin{align*}
		\Pr(\delta_i(k^+)=j)=\mathbb{E}[1_{\{\delta_i(k^+)=j\}}]
	\end{align*}
	therefore
	\begin{align*}
		&\frac{1}{M}\sum_{i=1}^{M}\Pr(\delta_i(k^+)\geq{\tt T}^*)\\=&\frac{1}{M}\sum_{i=1}^{M}\sum_{j\geq{\tt T}^*}\Pr(\delta_i(k^+)=j)\\
		=&\sum_{j\geq{\tt T}^*}\frac{1}{M}\sum_{i=1}^{M}\Pr(\delta_i(k^+)=j)\\=&\sum_{j\geq{\tt T}^*}\frac{1}{M}\sum_{i=1}^{M}\mathbb{E}[1_{\{\delta_i(k^+)=j\}}]\\
		=&\sum_{j\geq{\tt T}^*}\ell_j(k^+).
	\end{align*}
	Again, by the Cesaro Mean Lemma, 
	\begin{align}
		&\lim_{M\rightarrow\infty}J_{22}\nonumber\\
		\leq&\lim_{M\rightarrow\infty}\lim_{K\rightarrow\infty}\frac{1}{K}\sum_{k=1}^{K}\frac{1}{M}cM\left(\sum_{j\geq T^*}\ell_j(k^+)\right)\nonumber\\
		=&\lim_{M\rightarrow\infty}\frac{1}{M}cM\left(\sum_{j\geq T^*}\ell^*_j\right)\nonumber\\
		=&\lim_{M\rightarrow\infty}\frac{1}{M}cM(1-s_T^*)\nonumber\\
		=&0.
	\end{align}
	The last equality follows from \eqref{eq: lim_s_star} ($\lim_{M\to\infty}s_{T^*}=1$).
	Finally, summing $J_1$, $J_{21}$ and $J_{22}$, we find
	\begin{align*}
		\lim_{M\rightarrow\infty} \mathbb{E}[J^{SAT}(M)]=\frac{e}{2}.
	\end{align*}

	\section{Proof of Theorem~\ref{thm: general limit threshold}.}\label{App: {thm: general limit threshold}}
	Summing \eqref{eq: limit collection} on both sides, we have
	\begin{align}
		\sum_{m\geq1}a_m^*=\theta.\label{eq:totalsuma}
	\end{align}
	From the definition of the threshold in \eqref{eq: threshold}, ${\tt T}^*$ satisfies
	\begin{align}\label{eq: Appendix_Thresh2}
		{\tt T}^*=\max\left\{t|\sum_{m\geq t}a_m^*\geq \min\big(\theta, \frac{C^{\pi^{(1)}}}{M}\big)\right\}.
	\end{align}
	If $\theta\leq\frac{C^{\pi^{(1)}}}{M}$,  we have ${\tt T}^*=1$ by \eqref{eq:totalsuma}.
	If $\theta>\frac{C^{\pi^{(1)}}}{M}$, however, we have 
	\begin{align}
		\frac{C^{\pi^{(1)}}}{M}\leq \sum_{m\geq {\tt T}^*}a_m^*=&\sum_{m\geq 1}a_m^*-\sum_{m< {\tt T}^*}a_m^*\nonumber\\\stackrel{(a)}{=}&\theta-({\tt T}^*-1)\frac{\theta C^{\pi^{(1)}}}{M}\label{eq:findT}
	\end{align}
	where $(a)$ follows from \eqref{eq:totalsuma} and \eqref{eq: l_m+C3}.
	Using \eqref{eq:findT} and noting that ${\tt T^*}$ is  integer, we find
	\begin{align}
		{\tt T}^*=\left\lfloor \frac{M}{C^{\pi^{(1)}}}-\frac{1}{\theta}+1 \right\rfloor.
	\end{align}

	\section{Proof of Theorem~\ref{thm: theta=1, xi/2}.}\label{App: {thm: theta=1, xi/2}}
	The proof of Theorem~\ref{thm: theta=1, xi/2} is almost exactly the same as that of Theorem~\ref{thm: theta=1, e/2}. 
	After replacing the sum arrival rate of the channel, $e^{-1}$, by $C^{\pi^{(1)}}$, from Part~$1$, Part~$2$ and Part~$3$ in the proof of Theorem~\ref{thm: theta=1, e/2}, we have
	\begin{align}
		J_1=&\frac{1}{M}\mathbb{E}[w]=\frac{1-\theta}{M\theta}\\
		\lim_{M\rightarrow\infty}J_{21}=&\lim_{M\rightarrow\infty}\frac{1}{M}\sum_{j=1}^{{\tt T}^*-1}j\ell_j^{*+}\nonumber\\
		=&\lim_{M\rightarrow\infty}\frac{1}{M}\frac{\tt T^*(\tt T^*-1)}{2} \frac{C^{\pi^{(1)}}}{M}\nonumber\\
		=&\frac{1}{2C^{\pi^{(1)}}}
	\end{align}	
	and
	\begin{align}
		\lim_{M\rightarrow\infty}J_{22}\leq\lim_{M\rightarrow\infty}\lim_{K\rightarrow\infty}\frac{1}{K}\sum_{k=1}^{K}\frac{1}{M^2}\sum_{i=1}^McM\Pr(\delta_i(k^+)>{\tt T}^*).
	\end{align}
	From Part~$3$ in the proof of Theorem~\ref{thm: theta=1, e/2}, we knew that 
	the last inequality holds because $\mathbb{E}[\delta_i(k^+)|\delta_i(k^+)>{\tt T}^*]=O(M)$. This, however,  is not as oblivious here. 
	To show this, we first observe that $\delta_i(k)$ is increasing in $k$ until a delivery occurs. Now, note that $\delta_i(k^+)$ is upper bounded by ${\tt T}^*$ plus the peak age at the first delivery after time slot $k$. The peak age is  bounded by $X_i$ (the inter arrival time), which is $o(M)$ on average, plus delay $D_i$, whose expectation is upper bounded by a constant times $M$ as formulated in Lemma \ref{lem: D_i O(M)} below.   Therefore, by the counterpart of the proof of Theorem~\ref{thm: theta=1, e/2}, we have 
	\begin{align*}
		\lim_{M\rightarrow\infty}J_{22}=0
	\end{align*}
	and
	summing $J_1$, $J_{21}$ and $J_{22}$, 
	\begin{align*}
		\lim_{M\rightarrow\infty}\mathbb{E}[J^{GSAT}(M)]=\frac{1}{2C^{\pi^{(1)}}}.
	\end{align*}

	\begin{lemma}\label{lem: D_i O(M)}
		The expectation of delay, $\mathbb{E}[D_i]$, satisfies $$\mathbb{E}[D_i]\leq c'M$$ where $c'$ is a constant that depends on the employed transmission policy. 
	\end{lemma}
	\begin{proof}
		Recall that
		\begin{align*}
			\lim_{M\rightarrow\infty}C^{\pi^{(1)}}(M)=C^{\pi^{(1)}}.
		\end{align*}
		Denote the inter-delivery time for source $i$ by $I_i$. Thus
		the expected number of received packets from source $i$ from time slot $0$ to $K$ is $\frac{K}{\mathbb{E}[I_i]}$. Since $C^{\pi^{(1)}}(M)$ is the sum throughput, we have
		\begin{align*}
			C^{\pi^{(1)}}(M)=\lim_{K\rightarrow\infty}\frac{\sum_{i=1}^{M}\frac{K}{\mathbb{E}[I_i]}}{K}.
		\end{align*}
		Moreover, all nodes are statistically identical. Therefore,  
		\begin{align*}
			C^{\pi^{(1)}}(M)=\frac{M}{\mathbb{E}[I_i]}
		\end{align*}
		and
		\begin{align*}
			\mathbb{E}[I_i]=\frac{M}{C^{\pi^{(1)}}(M)}.
		\end{align*}
		Note that 
		\begin{align*}
			\mathbb{E}[D_i]\leq\mathbb{E}[I_i]
		\end{align*}
		and for any $\epsilon>0$, there exists a $N_0>0$ such that $C^{\pi^{(1)}}(M)\geq C^{\pi^{(1)}}-\epsilon$ for all $M\geq N_0$. Therefore, 
		\begin{align*}
			\mathbb{E}[D_i]\leq \frac{M}{C^{\pi^{(1)}}-\epsilon}\triangleq c'M.
		\end{align*}
	\end{proof}

	\ifCLASSOPTIONcaptionsoff
	\newpage
	\fi

	
	
	%
	\bibliographystyle{IEEEtran}
	\bibliography{references}

	%
	
	\begin{IEEEbiographynophoto}{Xingran Chen}
		received the BS degree in Statistics from the Central South University, Changsha, China, in 2015. He received the M.A. degree in Applied Mathematics and Computational Science in 2018 from University of Pennsylvania. He is currently working toward the Ph.D. degree in electrical and system engineering at the University of Pennsylvania. His research interests are in the modeling, analysis, and decision making of networked systems.
	\end{IEEEbiographynophoto}
	
	\begin{IEEEbiographynophoto}{Konstantinos Gatsis}
		(S'10--M'16) received the Ph.D. degree in electrical and systems engineering from the University of Pennsylvania, Philadelphia, PA, USA, in 2016. Currently, he is a Departmental Lecturer in the Department of Engineering Science, University of Oxford, UK. His research interests include control and optimization applied to cyber-physical systems.
		Dr. Gatsis received the 2014 O. Hugo Schuck Best Paper Award, the Student Best Paper Award at the 2013 American Control Conference, and was a Best Paper Award Finalist at the 2014 ACM/IEEE International Conference on Cyber-Physical Systems.
	\end{IEEEbiographynophoto}
	
	\begin{IEEEbiographynophoto}{Hamed Hassani}
		Hamed Hassani is currently an assistant professor of Electrical and Systems Engineering department as well as the Computer and Information Systems department, and the Statistics department at the University of Pennsylvania. Prior to that, he was a research fellow at Simons Institute for the Theory of Computing (UC Berkeley) affiliated with the program of Foundations of Machine Learning, and a post-doctoral researcher in the Institute of Machine Learning at ETH Zurich. He received a Ph.D. degree in Computer and Communication Sciences from EPFL, Lausanne. He is the recipient of the 2014 IEEE Information Theory Society Thomas M. Cover Dissertation Award, 2015 IEEE International Symposium on Information Theory Student Paper Award, 2017 Simons-Berkeley Fellowship, 2018 NSF-CRII Research Initiative Award, 2020 Air Force Office of Scientific Research (AFOSR) Young Investigator Award, 2020 National Science Foundation (NSF) CAREER Award, and 2020 Intel Rising Star award. He has recently been selected as the distinguished lecturer of the IEEE Information Theory Society in 2022-2023. 
	\end{IEEEbiographynophoto}
	
	\begin{IEEEbiographynophoto}{Shirin Saeedi Bidokhti}
		is an assistant professor in the Electrical and Systems Engineering Department at the University of Pennsylvania (UPenn). She received her M.Sc. and Ph.D. degrees in Computer and Communication Sciences from the Swiss Federal Institute
		of Technology (EPFL). Prior to joining UPenn, she was a postdoctoral scholar at Stanford University and the Technical University of Munich.
		She has also held short-term visiting positions at ETH Zurich,
		University of California at Los Angeles, and the Pennsylvania State
		University. Her research interests broadly include the design and
		analysis of network strategies that are scalable, practical, and
		efficient for use in Internet of Things (IoT) applications,
		information transfer on networks, as well as data compression
		techniques for big data. She is a recipient of the 2022 IT society Goldsmith
		lecturer award, 2021 NSF-CAREER award, 2019 NSF-CRII Research
		Initiative award and the prospective researcher and advanced postdoctoral fellowships from the Swiss National Science Foundation.
	\end{IEEEbiographynophoto}
	
	

\end{document}

%% file: fig-age1.tex
\begin{tikzpicture}[scale=0.7]
	\begin{axis}
		[axis lines=left,
		width=2.9in,
		height=2.5in,
		scale only axis,
		xlabel=$\theta$,
		ylabel=NAAoI,
		xmin=0,xmax=1,
		ymin=1,ymax=2,
		xtick={0.1,0.2,0.3,0.4,0.5,0.6,0.7,0.8,0.9,1},
		ytick={1,1.2,1.4,1.6,1.8,2},
		ymajorgrids=true,
		grid style=dashed,
		scatter/classes={
			a={mark=+, draw=black},
			b={mark=star, draw=black}
		}
		]

		\addplot[color=black, smooth,thick]
		coordinates{(0.002,1.6643)(0.05,1.4013)(0.1,1.3875)(0.15,1.3898)(0.2,1.3887)(0.25,1.3813)(0.3,1.3832)(0.35,1.3847)(0.4,1.3866)(0.45,1.3804)(0.5,1.3800)(0.55,1.3827)(0.6, 1.3768)(0.65,1.3764)(0.7, 1.3776)(0.75,1.3638)(0.8,1.3616)(0.85,1.3621)(0.9,1.3397)(0.95,1.2345)(1,1.0578)
		};
		
		\addplot[color=black, dashed,thick]
		coordinates{(0.002,1.6983)(0.05,1.4410)(0.1,1.4293)(0.15,1.4301)(0.2,1.4279)(0.25,1.4261)(0.3,1.4255)(0.35,1.4247)(0.4,1.4248)(0.45,1.4211)(0.5,1.4256)(0.55,1.4219)(0.6, 1.4232)(0.65,1.4212)(0.7, 1.4221)(0.75,1.4216)(0.8,1.4158)(0.85,1.4176)(0.9,1.4058)(0.95,1.3847)(1,1.3590)
		};
		
		\addplot[color=red, smooth,thick]
		coordinates{(0.01,1.8157)(0.05,1.4716)(0.1,1.4553)(0.15,1.4370)(0.2,1.4159)(0.25,1.4110)(0.3,1.4035)(0.35,1.4066)(0.4,1.4076)(0.45,1.4015)(0.5,1.4044)(0.55,1.4008)(0.6, 1.3889)(0.65,1.3780)(0.7, 1.3782)(0.75,1.3672)(0.8,1.3633)(0.85,1.3657)(0.9,1.3439)(0.95,1.2770)(1,1.0633)
		};
		
		\addplot[color=red, dashed,thick]
		coordinates{(0.01,1.7277)(0.05,1.4947)(0.1,1.4706)(0.15,1.4637)(0.2,1.4522)(0.25,1.4514)(0.3,1.4474)(0.35,1.4430)(0.4,1.4414)(0.45,1.4403)(0.5,1.4347)(0.55,1.4375)(0.6, 1.4420)(0.65,1.4317)(0.7, 1.4259)(0.75,1.4324)(0.8,1.4237)(0.85,1.4208)(0.9,1.4095)(0.95,1.3922)(1,1.3675)
		};

		\addplot[
		color=blue, smooth,thick]
		coordinates{(0.02,1.8459)(0.05,1.5713)(0.1,1.5012)(0.15,1.4866)(0.2,1.4533)(0.25,1.4425)(0.3,1.4368)(0.35,1.4347)(0.4,1.4345)(0.45,1.4339)(0.5,1.4276)(0.55,1.4168)(0.6,1.4141)(0.65,1.4172)(0.7,1.4180)(0.75,1.3953)(0.8,1.3844)(0.85,1.3738)(0.9,1.3463)(0.95,1.2945)(1,1.0721)
		};

		\addplot[
		color=blue, dashed,thick]
		coordinates{(0.02,1.8286)(0.05,1.5842)(0.1,1.5069)(0.15,1.4967)(0.2,1.4870)(0.25,1.4806)(0.3,1.4730)(0.35,1.4824)(0.4,1.4711)(0.45,1.4622)(0.5,1.4599)(0.55,1.4598)(0.6, 1.4512)(0.65,1.4500)(0.7, 1.4442)(0.75,1.4444)(0.8,1.4350)(0.85,1.4267)(0.9,1.4087)(0.95,1.3919)(1,1.3584)
		};

		\legend{Adaptive Thinning $M=500$, Stationary Thinning $M=500$, Adaptive Thinning $M=100$, Stationary Thinning $M=100$, Adaptive Thinning $M=50$, Stationary Thinning $M=50$}[legend pos= north east]
	\end{axis}
	\draw[shift={(0,1)},color=black] (2pt,-26.5pt) -- (2pt,-30pt) node[below] {$\frac{1}{M}$};
	
\end{tikzpicture}

%% file: fig-age5.tex
\begin{tikzpicture}[scale=0.7]
	\begin{axis}
		[axis lines=left,
		width=2.9in,
		height=2.5in,
		scale only axis,
		xlabel=$\theta$,
		ylabel=Probabilities,
		xmin=0.95, xmax=1,
		ymin=0, ymax=.5,
		xtick={},
		ytick={},
		ymajorgrids=true,
		legend style={at={(0.25,0.6)},anchor=west},
		grid style=dashed,
		scatter/classes={
			a={mark=+, draw=black},
			b={mark=star, draw=black}
		}
		]

		\addplot[color=black, mark=star,thick]
		coordinates{(0.95,0.3700)(0.955,0.3804)(0.96,0.3829)(0.965,0.3863)(0.97,0.3863)(0.975,0.3900)(0.98,0.4045)(0.985,0.4127)(0.99,0.4229)(0.995,0.4505)(1,0.4778)
		};
		
		\addplot[color=red, mark=triangle,thick]
		coordinates{(0.95,0.0222)(0.995,0.0215)(0.96,0.0219) (0.965,0.0216) (0.97,0.0226)(0.975,0.0218)(0.98,0.0215)(0.985,0.0217)(0.99,0.0216)(0.995,0.0223)(1,0.0221)
		};

		\legend{Prob. of successful transmission, Prob. of nodes being active}
	\end{axis}
\end{tikzpicture}

%% file: fig-age2_sep1.tex
\begin{tikzpicture}[scale=0.7]
	\begin{axis}
		[axis lines=left,
		width=2.9in,
		height=2.5in,
		scale only axis,
		xlabel=$\theta$,
		ylabel=NAAoI,
		xmin=0,xmax=1,
		ymin=0,ymax=5,
		xtick={0.1,0.2,0.3,0.4,0.5,0.6,0.7,0.8,0.9,1},
		ytick={0,1,2,3,4,5}, 
		ymajorgrids=true,
		legend style={at={(0.05, 0.78)},anchor=west},
		grid style=dashed,
		scatter/classes={
			a={mark=+, draw=black},
			b={mark=star, draw=black}
		}
		]

		
		\addplot[color=purple,  mark=star, thick]
		coordinates{(0.002,2.684991)(0.05,2.7149)(0.1,2.7149)(0.15,2.7103)(0.2,2.7129)(0.25,2.7083)(0.3,2.7164)(0.35,2.7149)(0.4,2.7158)(0.45,2.7156)(0.5,2.7122)(0.55,2.7143)(0.6,2.7145)(0.65,2.7151)(0.7,2.7103)(0.75,2.7157)(0.8,2.7141)(0.85,2.7139)(0.9,2.7136)(0.95,2.7122)(1,2.7163)
		};
		
		\addplot[color = purple, mark = square, thick]
		coordinates{(1, 1.5908)};
		
		\addplot[color=black,mark=triangle,thick]
		coordinates{(0.002,1.6983)(0.05,1.4410)(0.1,1.4293)(0.15,1.4301)(0.2,1.4279)(0.25,1.4261)(0.3,1.4255)(0.35,1.4247)(0.4,1.4248)(0.45,1.4211)(0.5,1.4256)(0.55,1.4219)(0.6, 1.4232)(0.65,1.4212)(0.7, 1.4221)(0.75,1.4216)(0.8,1.4158)(0.85,1.4176)(0.9,1.4058)(0.95,1.3847)(1,1.3590)
		};

		\addplot[color=red,
		mark=star,thick]
		coordinates{(0.002,1.6643)(0.05,1.4013)(0.1,1.3875)(0.15,1.3898)(0.2,1.3887)(0.25,1.3813)(0.3,1.3832)(0.35,1.3847)(0.4,1.3866)(0.45,1.3804)(0.5,1.3800)(0.55,1.3827)(0.6, 1.3768)(0.65,1.3784)(0.7, 1.3776)(0.75,1.3738)(0.8,1.3616)(0.85,1.3621)(0.9,1.3397)(0.95,1.2345)(1,1.0778)
		};

		\addplot[color=brown,thick]
		coordinates{(0.002,1.00)(0.05,0.88)(1,0.880)
		};

		\addplot[color=black, mark=o,thick]
		coordinates{(0.002,1.0064)(0.05,0.5023)(0.1,0.5013)(0.15,0.5008)(0.2,0.5005)(0.25,0.5002)(0.3,0.5000)(0.35,0.4999)(0.4,0.4998)(0.45,0.4996)(0.5,0.4995)(0.55,0.4995)(0.6, 0.4994)(0.65,0.4993)(0.7,0.4992)(0.75,0.4992)(0.8,0.4991)(0.85,0.4991)(0.9,0.4991)(0.95,0.4990)(1,0.4989)
		};

		\legend{Stationary Randomized \cite{S.K.Kaul-2017}, The ``LAZY'' policy in \cite{DCAEUOK2020}, Stationary Thinning with slotted ALOHA,  Adaptive Thinning with slotted ALOHA,   Lower Bounds \eqref{eq:lbMtheta} and \eqref{eq:lowerboundaloha} (distributed),Centralized Max-Weight \cite{IKEM2021}
		}
	\end{axis}
	\draw[shift={(0,1)},color=black] (0pt,-26.5pt) -- (0pt,-30pt) node[below] {$\frac{1}{M}$};
	
	\draw[shift={(-0.3,2.25)},color=black]  node[below] {$\frac{e}{2}$};
	
	\draw[shift={(-0.3,3.7)},color=black]  node[below] {$e$};

\end{tikzpicture}

%% file: fig-age2_sep2.tex
\begin{tikzpicture}[scale=0.7]
	\begin{axis}
		[axis lines=left,
		width=2.9in,
		height=2.5in,
		scale only axis,
		xlabel=$\theta$,
		ylabel=NAAoI,
		xmin=0,xmax=1,
		ymin=0,ymax=3,
		xtick={0.1,0.2,0.3,0.4,0.5,0.6,0.7,0.8,0.9,1},
		ytick={0,1,2,3,4,5}, 
		ymajorgrids=true,
		legend style={at={(0.01, 0.8)},anchor=west},
		grid style=dashed,
		scatter/classes={
			a={mark=+, draw=black},
			b={mark=star, draw=black}
		}
		]

		\addplot[color=black, mark = square, thick]
		coordinates{
			(0.1,1.0200)(0.15,1.0134)(0.2,1.0100)(0.25,1.0080)(0.3,1.0067)(0.35,1.0057)(0.4,1.0050)(0.45,1.0045)(0.5,1.0040)(0.55,1.0037)(0.6, 1.0034)(0.65,1.0031)(0.7,1.0029)(0.75,1.0027)(0.8,1.0025)(0.85,1.0024)(0.9,1.0022)(0.95,1.0021)(1,1.0020)
		};

		\addplot[color=green, mark=square,thick]
		coordinates{(0.002,1.0115)(0.05,0.5396)(0.1,0.5302)(0.15,0.5263)(0.2,0.5244)(0.25,0.5231)(0.3,0.5219)(0.35,0.5214)(0.4,0.5207)(0.45,0.5202)(0.5,0.5201)(0.55,0.5198)(0.6, 0.5195)(0.65,0.5193)(0.7,0.5191)(0.75,0.5189)(0.8,0.5187)(0.85,0.5185)(0.9,0.5184)(0.95,0.5183)(1,0.5183)
		};

		\addplot[color=brown, mark = square, thick]
		coordinates{(0.1,0.5110)(0.15,0.5110)(0.2,0.5110)(0.25,0.5110)(0.3,0.5110)(0.35,0.5110)(0.4,0.5110)(0.45,0.5110)(0.5,0.5110)(0.55,0.5110)(0.6,0.5110)(0.65,0.5110)(0.7,0.5110)(0.75,0.5110)(0.8,0.5110)(0.85,0.5110)(0.9,0.5110)(0.95,0.5110)(1,0.5110)
		};
		
		\addplot[color=purple,thick]
		coordinates{(0.002,1.00)(0.05,0.5010)(0.1,0.5010)(0.15,0.5010)(0.2,0.5010)(0.25,0.5010)(0.3,0.5010)(0.35,0.5010)(0.4,0.5010)(0.45,0.5010)(0.5,0.5010)(0.55,0.5010)(0.6, 0.5010)(0.65,0.5010)(0.7,0.5010)(0.75,0.5010)(0.8,0.5010)(0.85,0.5010)(0.9,0.5010)(0.95,0.5010)(1,0.5010)
		};
		
		\addplot[color=black, mark=o,thick]
		coordinates{(0.002,1.0064)(0.05,0.5023)(0.1,0.5013)(0.15,0.5008)(0.2,0.5005)(0.25,0.5002)(0.3,0.5000)(0.35,0.4999)(0.4,0.4998)(0.45,0.4996)(0.5,0.4995)(0.55,0.4995)(0.6, 0.4994)(0.65,0.4993)(0.7,0.4992)(0.75,0.4992)(0.8,0.4991)(0.85,0.4991)(0.9,0.4991)(0.95,0.4990)(1,0.4989)
		};

		\legend{The CSMA in \cite{AMMAAE2020}, Stationary Thinning with CSMA, Asymptotical Sleep-Wake scheme in \cite{AMBYSRSNBS2021}, Lower Bounds \eqref{eq:lowerboundaloha} and \eqref{eq:lowerboundcsma},
		Centralized Max-Weight \cite{IKEM2021}
		}
	\end{axis}
	\draw[shift={(0,1)},color=black] (0pt,-26.5pt) -- (0pt,-30pt) node[below] {$\frac{1}{M}$};

\end{tikzpicture}

%% file: fig-age6_sep1.tex
\begin{tikzpicture}[scale=0.7]
	\draw[shift={(7.1,-2pt)},color=black]  node[below] {$\frac{1}{M}$};
	\draw[shift={(-2pt,-2pt)},color=black]  node[below] {$0$};
	\draw[shift={(2.65,-2pt)},color=black]  node[below] {$\frac{1}{eM}$};
	\draw[shift={(2.3,-2pt)},color=gray,dashed] (0.3,0) -- (0.3,5.5) node[below] {}; hub9gbh657r

	\begin{axis}
		[axis lines=left,
		width=2.9in,
		height=2.5in,
		scale only axis,
		xlabel=$\theta$,
		ylabel=NAAoI,
		xmin=0, xmax=0.002,
		ymin=0, ymax=8,
		xticklabels=\empty,
		ymajorgrids=true,
		grid style=dashed,
		legend style={at={(0.01,.95)},anchor=west},
		scaled ticks=false, tick label style={/pgf/number format/fixed},
		scatter/classes={
			a={mark=+, draw=black},
			b={mark=star, draw=black}
		}
		]

		\addplot[
		color=blue,  mark=diamond,thick]
		coordinates{(0.00024525,8.1214)(0.00049051,4.0774)(0.00073576,2.7455)(0.00098101,2.7449)(0.0012,2.828455)(0.0015,3.121583)(0.0017,3.654413)(0.0018,4.205794)(0.0019,5.229598)(0.0020,7.144848)
		};
		
		\addplot[color=purple,  mark=star,thick]
		coordinates{(0.00024525,8.2473)(0.00049051,4.4665)(0.00073576,3.296686)(0.00098101,3.097862)(0.0012,2.910726)(0.0015,2.719564)(0.0017,2.695324)(0.0020,2.684991)
		};
		\addplot[color=black,
		mark=triangle,thick]
		coordinates{(0.00024525,8.1549)(0.00049051,4.0774)(0.00073576,2.7268)(0.00098101,2.2161)(0.0012,1.9514)(0.0015,1.8038)(0.0017,1.7406)(0.0020,1.6983)
		};

		\addplot[color=red, mark=star,thick]
		coordinates{(0.00024525,8.1549)(0.00049051,4.0774)(0.00073576,2.7263)(0.00098101,2.5984)(0.0012,2.2805)(0.0015,1.9658)(0.0017,1.7867)(0.0020,1.6643)
		};

		\addplot[color=black,thick,mark=o]
		coordinates{(0.00024525,8.1549)(0.00049051,4.0774)(0.00073576,2.7199)(0.00098101,2.0741)(0.0012,1.6681)(0.0015,1.3419)(0.0017,1.1152)(0.0020,1.0174)
		};

		\addplot[color=purple,thick]
		coordinates{(0.00024525,8.1549)(0.00049051,4.0774)(0.00073576,2.7183)(0.00098101,2.0387)(0.0012,1.6667)(0.0015,1.3333)(0.0017,1.1765)(0.0020,1.0000)
		};
		\addplot[color=brown,thick]
		coordinates{(0.00024525,    8.1549)
			(0.00049051 ,   4.0774)
			(0.00073576 ,   2.7183)
			(0.00098101  ,  2.0387)
			(0.0012  ,  1.6667)
			(0.0015  ,  1.3333)
			(0.0017 ,   1.1765)
			(0.0020  ,  1)
		};

		\legend{Slotted ALOHA (with unit buffer size), Stationary Randomized \cite{S.K.Kaul-2017}, Stationary Thinning with slotted ALOHA, Adaptive Thinning with slotted ALOHA, Centralized Max-Weight \cite{IKEM2021}, Lower Bound \eqref{eq:lbMtheta}}
	\end{axis}

\end{tikzpicture}

%% file: fig-age6_sep2.tex
\begin{tikzpicture}[scale=0.7]
	\draw[shift={(7.1,-2pt)},color=black]  node[below] {$\frac{1}{M}$};
	\draw[shift={(-2pt,-2pt)},color=black]  node[below] {$0$};
	\draw[shift={(2.65,-2pt)},color=black]  node[below] {$\frac{1}{eM}$};
	\draw[shift={(2.3,-2pt)},color=gray,dashed] (0.3,0) -- (0.3,5.5) node[below] {}; hub9gbh657r

	\begin{axis}
		[axis lines=left,
		width=2.9in,
		height=2.5in,
		scale only axis,
		xlabel=$\theta$,
		ylabel=NAAoI,
		xmin=0, xmax=0.002,
		ymin=0, ymax=8,
		xticklabels=\empty,
		ymajorgrids=true,
		grid style=dashed,
		legend style={at={(0.23,.85)},anchor=west},
		scaled ticks=false, tick label style={/pgf/number format/fixed},
		scatter/classes={
			a={mark=+, draw=black},
			b={mark=star, draw=black}
		}
		]

		\addplot[color=black,thick,mark=o]
		coordinates{(0.00024525,8.1549)(0.00049051,4.0774)(0.00073576,2.7199)(0.00098101,2.0741)(0.0012,1.6681)(0.0015,1.3419)(0.0017,1.1152)(0.0020,1.0174)
		};

		\addplot[color=green, mark=square,thick]
		coordinates{(0.00024525,8.0899)(0.00049051,4.0490)(0.00073576,2.7305)(0.00098101,2.0512)(0.0012,1.6822)(0.0015,1.3552)(0.0017,1.2051)(0.0020,1.0615)
		};
		
		\addplot[color=purple,thick]
		coordinates{(0.00024525,8.1549)(0.00049051,4.0774)(0.00073576,2.7183)(0.00098101,2.0387)(0.0012,1.6667)(0.0015,1.3333)(0.0017,1.1765)(0.0020,1.0000)
		};
		\addplot[color=brown,thick]
		coordinates{(0.00024525,    8.1549)
			(0.00049051 ,   4.0774)
			(0.00073576 ,   2.7183)
			(0.00098101  ,  2.0387)
			(0.0012  ,  1.6667)
			(0.0015  ,  1.3333)
			(0.0017 ,   1.1765)
			(0.0020  ,  1)
		};

		\legend{Centralized Max-Weight \cite{IKEM2021}, Stationary Thinning with CSMA, Lower Bound \eqref{eq:lbMtheta}}
	\end{axis}

\end{tikzpicture}